\PassOptionsToPackage{unicode}{hyperref}
\PassOptionsToPackage{hyphens}{url}
\documentclass[11pt, ]{article}

\usepackage{mathtools}
\usepackage{amsmath}
\usepackage{amsthm}
\usepackage{amssymb}
\usepackage[italicdiff]{physics}
\mathtoolsset{showonlyrefs}

\usepackage{iftex}
\ifPDFTeX
  \usepackage[OT1,T1]{fontenc}
  \usepackage[utf8]{inputenc}
  \usepackage{textcomp} 
    \usepackage[p,osf,swashQ]{cochineal}
  \usepackage[cochineal,vvarbb]{newtxmath}
      \usepackage[scale=0.95]{biolinum}
    \usepackage[scale=0.95,varl]{inconsolata}
\else 
  \usepackage[scale=0.95,varl]{inconsolata}
  \usepackage{newpxtext}
  \usepackage{mathpazo}
    \usepackage[scale=0.95]{biolinum}
  \fi
\ifLuaTeX
  \usepackage{selnolig}  
\fi
\IfFileExists{microtype.sty}{
  \usepackage[]{microtype}
  \UseMicrotypeSet[protrusion]{basicmath} 
}{}

\allowdisplaybreaks
\usepackage[dvipsnames,svgnames,x11names]{xcolor}
\usepackage[lmargin=1in,rmargin=1in,tmargin=1in,bmargin=1in]{geometry}
\usepackage[format=plain,
  labelfont={bf,sf,small,singlespacing},
  textfont={sf,small,singlespacing},
  justification=justified,
  margin=0.25in]{caption}

\usepackage{sectsty}
\usepackage[compact]{titlesec}
\makeatletter
\newcommand\@shorttitle{}
\newcommand\shorttitle[1]{\renewcommand\@shorttitle{#1}}
\usepackage{fancyhdr}
\fancyheadoffset{0pt}
\makeatother
\renewenvironment{abstract}{
  \centerline
  {\large\sffamily\bfseries Abstract}\vspace{-1em}
  \begin{quote}\small
}{
  \end{quote}
}

\makeatletter
\newcommand{\assumplabel}[2]{%
   \protected@write \@auxout {}{\string\newlabel{#1}{{#2}{\thepage}{#2}{#1}{}}}%
   \hypertarget{#1}{#2}%
}
\makeatother
\newenvironment{assump}[2][]{%
\par\medskip%
\def\tempa{}\def\tempb{#1}%
\lowercase{\def\lblkey{a-#2}}%
\textbf{Assumption \assumplabel{\lblkey}{#2}}%
\ifx\tempa\tempb\else~(#1)\fi%
\textbf{.}\itshape
}{\normalfont}

\providecommand{\tightlist}{%
  \setlength{\itemsep}{0pt}\setlength{\parskip}{0pt}}\usepackage{longtable,booktabs,array}
\usepackage{calc} 
\usepackage{etoolbox}
\makeatletter
\patchcmd\longtable{\par}{\if@noskipsec\mbox{}\fi\par}{}{}
\makeatother
\IfFileExists{footnotehyper.sty}{\usepackage{footnotehyper}}{\usepackage{footnote}}
\makesavenoteenv{longtable}
\usepackage{graphicx}
\makeatletter
\def\maxwidth{\ifdim\Gin@nat@width>\linewidth\linewidth\else\Gin@nat@width\fi}
\def\maxheight{\ifdim\Gin@nat@height>\textheight\textheight\else\Gin@nat@height\fi}
\makeatother
\setkeys{Gin}{width=\maxwidth,height=\maxheight,keepaspectratio}
\makeatletter
\def\fps@figure{htbp}
\makeatother

\usepackage{cancel}
\usepackage{thmtools}
\usepackage{tikz}
\usetikzlibrary{bayesnet}

\allowdisplaybreaks

\newcommand{\proj}{\mathrm{proj}}

\providecommand{\norm}[1]{\lVert#1\rVert}
\newcommand{\Categorical}{\mathrm{Categorical}}
\newcommand{\notindep}{\mathbin{\perp\!\!\!\!\not\!\:\perp}}

\newcommand{\cS}{\mathcal{S}}
\newcommand{\cR}{\mathcal{R}}
\newcommand{\cG}{\mathcal{G}}

\newcommand{\cX}{\mathcal{X}}
\newcommand{\cY}{\mathcal{Y}}
\newcommand{\cW}{\mathcal{W}}
\newcommand{\cI}{\mathcal{I}}

\theoremstyle{plain}
\newtheorem{prop}{Proposition}[section]

\newtheorem{theorem}[prop]{Theorem}

\newcommand\numberthis{\addtocounter{equation}{1}\tag{\theequation}}
\usepackage[labelformat=simple]{subcaption}

\newcommand{\R}{\ensuremath{\mathbb{R}}}

\newcommand{\half}{\frac{1}{2}}

\newcommand{\eps}{\varepsilon} 

\DeclareMathOperator*{\argmax}{arg\,max}

\DeclareMathOperator{\E}{\mathbb{E}}
\let\Pr\undefined
\DeclareMathOperator{\Pr}{\mathbb{P}}
\newcommand{\indep}{\mathbin{\perp\!\!\!\!\!\:\perp}} 

\newcommand{\ind}{\mathbf{1}}
\newcommand{\iid}{\mathrel{\stackrel{iid}{\sim}}} 

\newcommand{\cvas}{\xrightarrow{\;\!a.s.\:\!}}

\DeclareMathOperator{\Cov}{Cov}

\newcommand{\Dirichlet}{\mathrm{Dirichlet}\qty}

\newcommand{\Norm}{\mathcal{N}\qty}

\usepackage{booktabs}
\usepackage{longtable}
\usepackage{array}
\usepackage{multirow}
\usepackage{wrapfig}
\usepackage{float}
\usepackage{colortbl}
\usepackage{pdflscape}
\usepackage{tabu}
\usepackage{threeparttable}
\usepackage{threeparttablex}
\usepackage[normalem]{ulem}
\usepackage{makecell}
\usepackage{xcolor}
\makeatletter
\@ifpackageloaded{caption}{}{\usepackage{caption}}
\AtBeginDocument{%
\ifdefined\contentsname
  \renewcommand*\contentsname{Table of contents}
\else
  \newcommand\contentsname{Table of contents}
\fi
\ifdefined\listfigurename
  \renewcommand*\listfigurename{List of Figures}
\else
  \newcommand\listfigurename{List of Figures}
\fi
\ifdefined\listtablename
  \renewcommand*\listtablename{List of Tables}
\else
  \newcommand\listtablename{List of Tables}
\fi
\ifdefined\figurename
  \renewcommand*\figurename{Figure}
\else
  \newcommand\figurename{Figure}
\fi
\ifdefined\tablename
  \renewcommand*\tablename{Table}
\else
  \newcommand\tablename{Table}
\fi
}
\@ifpackageloaded{float}{}{\usepackage{float}}
\floatstyle{ruled}
\@ifundefined{c@chapter}{\newfloat{codelisting}{h}{lop}}{\newfloat{codelisting}{h}{lop}[chapter]}
\floatname{codelisting}{Listing}

\makeatother
\makeatletter
\@ifpackageloaded{tikz}{}{\usepackage{tikz}}
\makeatother
\makeatletter
\@ifpackageloaded{caption}{}{\usepackage{caption}}
\@ifpackageloaded{subcaption}{}{\usepackage{subcaption}}
\makeatother

\usepackage[]{natbib}
\newenvironment{CSLReferences}[2]{
\bibliography{references.bib}
\clearpage
}{}

\usepackage{hyperref}
\usepackage{url}
\hypersetup{
  pdftitle={Estimating Racial Disparities When Race is Not Observed},
  pdfkeywords={race imputation, BISG, ecological inference, instrumental
variable, proxy variable},
  colorlinks=true,
  linkcolor={black},
  filecolor={Maroon},
  citecolor={VioletRed4},
  urlcolor={DodgerBlue4},
  pdfcreator={LaTeX via pandoc}}

\title{\sffamily\bfseries\huge\parfillskip=0pt
\rightskip=0pt plus .5\textwidth
\leftskip=0pt plus .5\textwidth
\emergencystretch=.3\textwidth Estimating Racial Disparities When Race
is Not Observed}
\shorttitle{Estimating Racial Disparities When Race is Not Observed}
\author{\textbf{Cory McCartan}
\\Center for Data Science%
\\New York University%
\vspace{2pt}
 \and \textbf{Robin Fisher}
\\Office of Tax Analysis%
\\U.S. Department of the Treasury%
\vspace{2pt}
 \and \textbf{Jacob Goldin}
\\Law School%
\\University of Chicago%
\vspace{2pt}
 \and \textbf{Daniel E. Ho}
\\Department of Political Science%
\vspace{2pt}
\\Department of Computer Science%
\vspace{2pt}
\\Law School%
\\Stanford University%
\vspace{2pt}
 \and \textbf{Kosuke Imai}\footnote{
To whom correspondence should be addressed.
Email: \texttt{\href{mailto:imai@harvard.edu}{imai@harvard.edu}}.
Website: \url{https://imai.fas.harvard.edu}.
We thank Bruce Willsie, CEO of L2, Inc., for providing us with the
geocoded voter file we use in this paper. We also thank Hiroto
Katsumata, Soichiro Yamauchi, and an anonymous reviewer of the Alexander
and Diviya Magaro Peer Pre-Review Program for useful feedback. The views
expressed in this paper are those of the authors and do not necessarily
represent the views of the US Treasury Department. Any taxpayer data
used in this research was kept in a secured IRS data repository, and all
results have been reviewed to ensure that no confidential information is
disclosed.}
\\Department of Government%
\vspace{2pt}
\\Department of Statistics%
\\Harvard University%
\vspace{2pt}
 }
\date{April 16, 2024}

\begin{document}
\allsectionsfont{\sffamily}

\maketitle

\begin{abstract}
The estimation of racial disparities in various fields is often hampered
by the lack of individual-level racial information. In many cases, the
law prohibits the collection of such information to prevent direct
racial discrimination. As a result, analysts have frequently adopted
Bayesian Improved Surname Geocoding (BISG) and its variants, which
combine individual names and addresses with Census data to predict race.
Unfortunately, the residuals of BISG are often correlated with the
outcomes of interest, generally attenuating estimates of racial
disparities. To correct this bias, we propose an alternative
identification strategy under the assumption that surname is
conditionally independent of the outcome given (unobserved) race,
residence location, and other observed characteristics. We introduce a
new class of models, Bayesian Instrumental Regression for Disparity
Estimation (BIRDiE), that take BISG probabilities as inputs and produce
racial disparity estimates by using surnames as an instrumental variable
for race. Our estimation method is scalable, making it possible to
analyze large-scale administrative data. We also show how to address
potential violations of the key identification assumptions. A validation
study based on the North Carolina voter file shows that BISG+BIRDiE
reduces error by up to 84\% when estimating racial differences in party
registration. Finally, we apply the proposed methodology to estimate
racial differences in who benefits from the home mortgage interest
deduction using individual-level tax data from the U.S. Internal Revenue
Service. Open-source software is available which implements the proposed
methodology.
\end{abstract}

\textbf{\textit{Keywords}}\quad race
imputation~\textbullet~BISG~\textbullet~ecological
inference~\textbullet~instrumental variable~\textbullet~proxy variable


\newpage

\section{Introduction}\label{introduction}

The identification and estimation of racial disparities is of paramount
importance to researchers, policymakers and organizations in a variety
of areas including public health
\citep{van2003paved, williams2005social}, employment
\citep{conway1983reverse, greene1984reverse}, voting
\citep{gay2001effect, hajnal2005turnout, barreto2007isi}, criminal
justice
\citep{berk2021fairness, chouldechova2017fair, dressel2018accuracy},
economic policy and taxation
\citep{brown2022whiteness, elzayn2023measuring}, housing
\citep{kermani2021racial}, lending \citep{chen2018fair}, and technology
and fairness \citep{alao2021meta}. Within the U.S. government, efforts
to identify and remedy racial disparities have taken on greater urgency
with the recent issuance of Executive Order 13985, which in part directs
agencies to conduct equity assessments by developing appropriate
methodology.

In many of these areas, however, racial information is not available at
the individual level. The unavailability of individual racial
information makes it impossible for analysts to simply tabulate
variables of interest against race to identify disparities among
different racial groups. In fact, in some areas, the law explicitly
prohibits the collection of racial information even as it demands fair
treatment on the basis of race (see, e.g., the U.S. Equal Credit
Opportunity Act). This creates a dilemma for organizations who wish to
measure possible disparities in order to monitor the fairness of their
decision-making or service provision.

To estimate racial disparities without individual-level racial data,
some researchers have turned to ecological inference methods
\citep{goodman1953ecological, king1997solution, king2004ecological, wakefield2004ecological, imai2008bayesian, greiner2009r}.
These methods, however, require strong assumptions, which can be
difficult to verify and may provide misleading results
\citep{cho2008ecological}. Additionally, they all rely on accurate
marginal information about race, which may not always be available.

Where the analysis of racial disparities involves large-scale
administrative data, many analysts have adopted Bayesian Improved
Surname Geocoding (BISG), which generates individual probabilities of
belonging to different racial groups using Bayes' rule applied to last
names and geographic location
\citep{fiscella2006bisg, elliott2008bisg, imai2016improving}. BISG
leverages residential racial segregation and the association between
self-reported race and surname to produce generally accurate and
calibrated predictions of self-reported individual race
\citep{kenny2021das, deluca2022validating}.

Much attention has been given to ways of increasing the accuracy of BISG
and related methods for race \emph{prediction}
\citep{voicu2018, zrp, argyle2022misclassification, imai2022addressing, decter2022, greengard2023bisg}.
Unfortunately, broadly accurate BISG racial prediction alone does not
guarantee the unbiased estimation of racial \emph{disparities}, which is
the ultimate goal of most analysts. To estimate disparities, BISG
probabilities (or any other racial predictions) must be combined with
information on the outcome variable for which the disparities are of
interest. But the most common techniques for doing so are known to be
biased when race is correlated with the outcome even after controlling
on name and location
\citep{chen2019fairness, argyle2022misclassification}. These approaches
include \emph{weighting} the outcome variable by the BISG probabilities,
and \emph{thresholding} the BISG probabilities to produce point
estimates of individual race (e.g., imputing ``Black'' as the race for
an individual with a 61\% probability of being Black according to BISG).

In fact, these methods often \emph{underestimate} the true magnitude of
racial disparities, which is problematic for policymakers and analysts
who aim to identify and reduce these disparities. As formally discussed
in Section~\ref{sec-id}, the standard methods of racial disparity
estimation based on BISG predictions implicitly require individuals'
race to be conditionally independent of the outcome given their
residence location, surnames, and other observable attributes. This key
identification assumption, however, is unlikely to hold because race
affects many aspects of society even after accounting for residence
location, surnames, and other observable attributes. Some researchers
have noted the implausibility of this assumption and have proposed
methods to address it, but these alternative approaches are more general
than the racial disparity setting and require additional data such as a
partially labeled subset \citep{fong2021machine}. Others have advocated
for partial identification strategies
\citep{kallus2021assessing, elzayn2023measuring}.

To address this challenge, in Section~\ref{sec-est}, we propose an
alternative identification strategy. Specifically, we assume that the
outcome is conditionally independent of surname given (unobserved)
individual's race, residence location, and other observed attributes.
This assumption is a type of exclusion restriction where surname serves
as an instrumental variable for unobserved race. It implies that for two
individuals who live in the same area, belong to the same racial group,
and share the observable attributes, their surnames have no predictive
power of the outcome. Somewhat counter-intuitively, the
high-dimensionality of surnames aids rather than hinders identification
because it provides a large number of instruments. We argue that this
new identification assumption is more credible than the commonly invoked
assumption unless surname is directly used to determine the outcome of
interest (i.e., name-based discrimination).

Leveraging this identification strategy, in Section~\ref{sec-birdie} we
introduce a new class of models, Bayesian Instrumental Regression for
Disparity Estimation (BIRDiE), that accurately estimates racial
disparities using BISG probabilities. Beyond accuracy, BIRDiE improves
on standard methodology in several ways:

\begin{itemize}
\tightlist
\item
  BIRDiE includes built-in flexibility for researchers to make
  problem-specific modeling choices (Section~\ref{sec-birdie}).
\item
  BIRDiE can be fit with an EM algorithm that can scale to hundreds of
  thousands or millions of observations (Section~\ref{sec-compute}).
\item
  BIRDiE produces updated BISG probabilities that incorporate the
  outcome variable and are likely to be more accurate than the BISG
  probabilities based only on surnames and geolocation
  (Section~\ref{sec-update-bisg}).
\item
  BIRDiE can be used iteratively to condition on additional variables
  whose distribution by race is not known \emph{a priori}
  (Section~\ref{sec-addlcov}). For example, party identification can be
  estimated by race \emph{and} turnout.
\end{itemize}

Finally, in Section~\ref{sec-sens} we address potential violations of
the key identification assumption, such as the one caused by overly
coarse racial categories, by exploiting auxiliary information about the
relations between names and more specific ethnic groups. All of the
proposed methodology is implemented in a computationally efficient
open-source software package, \texttt{birdie}, that is made available
with the paper. The software and accompanying documentation are
available at \url{https://corymccartan.com/birdie/}.

In Section~\ref{sec-valid}, we validate the proposed methodology using
real-world data taken from the voter file in North Carolina, where
self-reported individual-race is observed and can be used to construct
the ground-truth of racial disparities. BIRDiE substantially outperforms
existing estimators across different error measures and multiple levels
of geolocation specificity. For example, the most popular existing
BISG-only disparity estimator pegs the gap at Democratic party
registration between White and Black voters at 24.1 percentage points,
while the actual gap is 54.6 percentage points---more than double. Our
preferred BIRDiE model using the same BISG probabilities yields an
estimate of 48.5 percentage points. This represents about a 80\%
reduction in bias.

In Section~\ref{sec-irs}, we apply BIRDiE to large-scale administrative
tax data from the U.S. Internal Revenue Service (IRS), which does not
collect individual taxpayers' racial information. We produce the novel
estimates of the distribution by race in who claims the home mortgage
interest deduction---a question that has been largely hampered by a lack
of administrative tax data with taxpayers' racial information. The
results show that there exists a substantial degree of racial disparity
with many fewer Black and Hispanic filers claiming the HMID than White
and Asian filers. We find that the racial gaps in homeownership rates
alone cannot explain this disparity. Section~\ref{sec-disc} gives
concluding remarks.

\section{Bias of the Standard Methodology}\label{sec-id}

In this section, we review the assumptions of the standard BISG-based
methodology for estimating racial disparities when individual race is
not observed. We show that the racial disparity estimates based on the
standard methodology are biased unless the outcome variable is
independent of race given surname, residence location, and other
observed covariates. We argue that this assumption is likely to be
violated given the significant role race plays in many aspects of our
society.

\subsection{Setup and BISG Procedure}\label{setup-and-bisg-procedure}

Suppose that we have an i.i.d. sample of \(N\) individuals from a
population of infinite size. For each individual \(i=1,2,\ldots,N\), we
define a tuple \((Y_i, R_i, G_i, X_i, S_i)\), where \(Y_i\in\cY\) is the
outcome of interest for individual \(i\), \(R_i\in \cR\) is the
(unobserved) race of the individual, \(G_i\in\cG\) is the (geo)location
of the individual's residence, \(X_i\in\cX\) are other observed
characteristics, and \(S_i\in\cS\) is the individual's surname. When we
are not referring to a particular individual, we will drop the
subscripts for notational simplicity. Note that individual race is
unobservable but all other variables are assumed to be observed. The
availability of particular (or any) \(X\) is not required for either the
standard or proposed methodology.

We assume throughout that these variables are discrete, taking a finite
set of values, i.e., \(|\cY|\), \(|\cR|\), \(|\cG|\), \(|\cX|\), and
\(|\cS|\) are constants. Note that typically \(S\) is high-dimensional
as there exist a large number of unique surnames. In practice, residence
location \(G\) is also discrete, since joint information about location,
race, and other variables is generally only available down to the Census
block level. For the sake of simplicity, we assume that the outcome
variable \(Y\) is also discrete, though it is possible to extend the
standard and proposed methodologies to continuous outcome variables.

Typically, BISG relies on data from the decennial Census or the American
Community Survey (ACS), which provide information on the joint
distribution of \(R\) and \(G\) (and any other covariates \(X\), such as
gender or age). It then combines this information with data from the
Census Bureau's surname tables \citep{censusnames}, which provide
information on the joint distribution of \(R\) and \(S\). We summarize
this set of information from the Census by two conditional
probabilities, \(\vb q_{GX|R}\) and \(\vb q_{S|R}\), and one marginal
probability, \(\vb q_R\).

The BISG estimator of the probability that individual \(i\) belongs to
race \(r \in \cR\) can then be written as
\citep{fiscella2006bisg, elliott2008bisg}
\begin{equation} \label{eq:pr-bisg}
    \hat{P}_{ir} \coloneq \frac{q_{G_iX_i|r}\, q_{S_i|r}\, q_r}{
        \sum_{r^\prime\in\cR}q_{G_iX_i|r^\prime}\, q_{S_i|r^\prime}\, q_{r^\prime}},
\end{equation} where, for example, \(q_{G_iX_i|r}\) indicates the
estimated conditional probability of residence location \(G_i\) and
covariates \(X_i\) given race \(r\), taken from the Census table
\(\vb q_{GX|R}\).

The BISG estimator relies on two key assumptions. The first is that the
Census tables reflect the true population distributions of \(R\), \(S\),
\(G\), and \(X\).

\begin{assump}[Data accuracy]{ACC}
    For all $i$, 
    \begin{align*}
        \Pr(R_i=r) &= q_r \\
        \Pr(S_i=s\mid R_i = r) &= q_{s|r} \\
        \Pr(G_i=g, X_i=x\mid R_i = r) &= q_{gx|r}
    \end{align*}
\end{assump}

Despite the best efforts of the Census Bureau, Assumption \ref{a-acc}
may never hold exactly in practice. The decennial census has intrinsic
error, including undercounting minority racial groups
\citep{census2022undercount, censuscount, racecounts}, as well as error
introduced by privacy-preserving mechanisms
\citep{abowd2020, mccartan2023nmf}. And because of births, deaths, and
moves, census data are often out-of-date from the moment of publication.
These errors have led further extensions of the BISG estimator to
account for some of this measurement error \citep{imai2022addressing},
which can help with accuracy for smaller racial groups. The plausibility
of Assumption \ref{a-acc} is stretched even further when the study
population is a subset of the whole U.S. population, and so is not
covered by national census data. In these cases, analysts should set
\(\vb q_R\) to the known or estimated marginal racial distribution in
the study population, rather than the national racial distribution, when
this margin in known. It may be more plausible then to assume that the
conditional distributions \(\Pr(S\mid R)\) and \(\Pr(G,X\mid R)\) match
the census distributions, even if \(\Pr(R)\) does not
\citep{rosenman2023firstname}. \citet{greengard2023bisg} take this
approach a step further by by raking BISG probabilities to all known
margins, further improving calibration.

The second assumption required by BISG is the following conditional
independence relation between an individual's surname and residence
location (as well as other characteristics) given their unobserved race.

\begin{assump}[Conditional independence of name and other proxy variables]{CI-SG}
    For all $i$, $$S_i \indep \{G_i, X_i\}\mid R_i.$$
\end{assump}

Assumption \ref{a-ci-sg} implies, for example, that once we know an
individual is White, knowing their surname is Smith tells us nothing
about their residence location and other observed characteristics.
Although this assumption appears to be reasonable, the lack of
granularity in the coding of race may lead to its violation. For
example, people with Chinese, Indian, Filipino, Vietnamese, Korean, or
Japanese are all coded as one racial group ``Asian'' in the census.
These groups, however, have varying sets of surnames and have different
demographic and geographic distributions. For instance, unlike the Smith
example, knowing that an Asian individual's surname is Gupta makes it
more likely that they have a higher income and live in the Eastern U.S
\citep{budiman2019key}.

Even though Assumptions \ref{a-ci-sg} and \ref{a-acc} may not hold
exactly, researchers find that BISG produces accurate and generally
well-calibrated estimates in practice
\citep{imai2016improving, zhang2018bisg, kenny2021das, deluca2022validating}.
We observe this pattern as well in the validation study in
Section~\ref{sec-valid}.

Under Assumption \ref{a-ci-sg}, by Bayes' Rule, \[
    \Pr(R_i=r\mid G_i,X_i,S_i)
    \propto \Pr(G_i, X_i\mid R_i=r)\Pr(S_i\mid R_i=r)\Pr(R_i=r).
\] This justifies the estimator given in Equation \eqref{eq:pr-bisg},
and provides us with the following immediate result.

\begin{prop}[Accuracy of BISG] \label{p:bisg}
Under Assumptions \ref{a-ci-sg} and \ref{a-acc}, the BISG estimator produces correct probabilities. That is, we have
$\hat P_{ir} = \Pr(R_i=r\mid G_i,X_i,S_i)$.
\end{prop}

New methods continue to be developed that improve the calibration of
BISG probabilities, including some machine learning methods based on
labeled data
\citep{zrp, imai2022addressing, argyle2022misclassification, decter2022, greengard2023bisg, cheng2023redundant}.
Fundamentally, these approaches focus on building a more accurate model
for \(R\mid G,X,S\) at the individual level.

\subsection{Bias of Racial Disparity Estimates based on BISG
Probabilities}\label{bias-of-racial-disparity-estimates-based-on-bisg-probabilities}

To estimate racial disparities, BISG probabilities (or other racial
predictions) must be combined with the outcome variable. There are
several common ways researchers do this.

The most frequent is perhaps the \emph{thresholding} or
\emph{classification} estimator, which deterministically assigns
individuals to a predicted racial category based on the BISG estimates
\(\vb{\hat P}_i\) (either the largest \(\hat P_{ir}\) or the one which
exceeds a predetermined threshold). Estimates of \(\Pr(Y=y\mid R=r)\)
are then obtained by tabulating the data by these assigned categories.

Another common approach, which attempts to capture the uncertainty
inherent in race prediction, is the following \emph{weighting}
estimator: \[
    \hat\mu^{(\text{wtd})}_{Y|R}(y\mid r) 
        = \frac{\sum_{i=1}^N \ind\{Y_i=y\}\hat P_{ir}}{\sum_{i=1}^N \hat P_{ir}}.
\]

Unfortunately, accurate and calibrated estimates of individual race
predictions alone, unless they are perfect, are not sufficient for
unbiased estimation of racial disparities using these standard
methodologies.

This should come as no surprise for the thresholding estimator, since it
does not take into account prediction uncertainty in the BISG
probabilities. This is akin to ignoring measurement errors in the
covariates of a regression, something which has long been known to lead
to biased coefficient estimates. Unlike the classical
errors-in-variables setting, however, the bias of the thresholding
estimator is not consistently in the same direction, making it hard to
reason about \citep{chen2019fairness}.

But, the weighting estimator is also biased because the prediction error
of race probabilities may be correlated with the outcome variable of
interest. Fortunately, unlike the threshold estimator, it is easier to
understand the nature of this bias. \citet{chen2019fairness} show that
the asymptotic bias of the weighting estimator is controlled by the
residual correlation of \(Y\) and \(R\) after adjusting for \(G\),
\(X\), and \(S\). We reproduce this result here.

\begin{theorem}[Theorem 3.1 of \citealt{chen2019fairness}] \label{thm:wt-bias}
If race is binary (so $\cR=\{0,1\}$), then as $N\to\infty$, \[
    \hat\mu^{(\text{wtd})}_{Y|R}(y\mid r) - \Pr(Y=y\mid R=r)
    \cvas -\frac{\E[\Cov(\ind\{Y=y\}, \ind\{R=r\}\mid G,X,S)]}{\Pr(R=r)}.
\]
\end{theorem}

This result implies that when the BISG residuals
\(\ind\{R=r\}-\Pr(R=r\mid G,X,S)\) are correlated with the outcome,
estimates will be biased, even with infinite data. In fact, the
weighting estimator will often underestimate the magnitude of a
disparity, as the following corollary shows. For instance, in measuring
disparities in loan approval (\(Y\)), if Blacks are less likely to be
approved for loans across all locations and surnames than Whites, then
the weighting estimator would understate the resulting overall
White-Black disparity in loan approval rates.

\begin{restatable}[Underestimation of racial disparity]{corollary}{corunder} \label{cor:under}
Let $y\in\cY$. If race is binary (so $\cR=\{0,1\}$), and 
$\Pr(Y=y\mid R=1, G=g, X=x, S=s)>\Pr(Y=y\mid R=0, G=g, X=x, S=s)$ for all $g\in\cG$, $x\in\cX$, and $s\in\cS$,
then \[
    \hat\mu^{(\text{wtd})}_{Y|R}(y\mid 1) - \hat\mu^{(\text{wtd})}_{Y|R}(y\mid 0)
    < \Pr(Y=y\mid R=1)-\Pr(Y=y\mid R=0).
\]
\end{restatable}

Conversely, Theorem \ref{thm:wt-bias} implies that conditional
independence between individual's race and outcome given their surname,
residence location, and other characteristics is sufficient to eliminate
the asymptotic bias of the weighting estimator.

\begin{assump}[Conditional independence of outcome and race]{CI-YR}
For all $i$, $$Y_i\indep R_i\mid G_i,X_i,S_i.$$
\end{assump}

Figure~\ref{fig-dag-yr} shows a causal directed acyclic graph (DAG) that
satisfies Assumption \ref{a-ci-yr} as well as Assumption \ref{a-ci-sg}.
The dashed node border for \(R\) represents the fact that race is
unobserved. The causal structure in Figure~\ref{fig-dag-yr} implies the
conditional independence relation \(Y\indep R\mid G,X,S\), because all
paths from \(R\) to \(Y\) are blocked by \(G\), \(X\), or \(S\). The key
causal assumption of this DAG is that the effect of race \(R\) on the
outcome \(Y\) must be entirely mediated by surname \(S\), residence
location \(G\), and other observed characteristics \(X\). This type of
exclusion restriction may not be credible in many practical settings
because race can affect the outcome through so many factors, biasing the
weighting estimator.

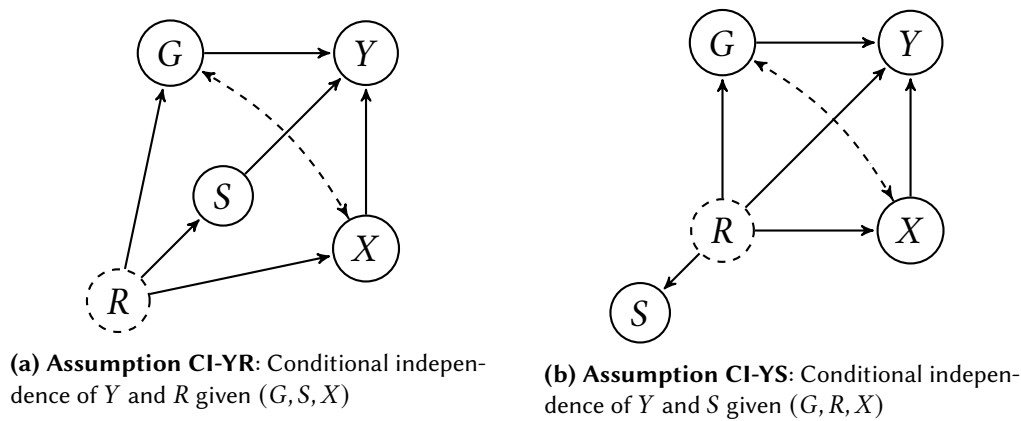
\begin{figure}

\begin{minipage}{0.10\linewidth}
~\end{minipage}%
\begin{minipage}{0.38\linewidth}

\centering{

\tikzstyle{main node}=[circle,draw,font=\sffamily\Large\bfseries]
\tikzstyle{sub node}=[circle,draw,dashed,font=\sffamily\Large\bfseries]
\begin{tikzpicture}[->,>=stealth',shorten >=1pt,auto,node distance=2.6cm,thick]

\node[main node] (G) {$G$};
\node[main node] (Y) [right of=G] {$Y$};
\node[main node] (S) [below left=1.3cm and 1.3cm of Y] {$S$};
\node[sub node] (R) [below left=0.8cm and 0.8cm of S] {$R$};
\node[main node] (X) [below of=Y] {$X$};

\path[every node/.style={font=\sffamily\small}]
(R) edge node  {} (G)
(R) edge node  {} (X)
(S) edge node  {} (Y)
(R) edge node  {} (S)
(G) edge [dash pattern=on 0.0pt off 20pt, bend left=15] node  {} (X)
(X) edge [dashed, bend right=15] node  {} (G)
(G) edge node  {} (Y)
(X) edge node  {} (Y);

\end{tikzpicture}

}

\subcaption{\label{fig-dag-yr}\textbf{Assumption \ref{a-ci-yr}}:
Conditional independence of \(Y\) and \(R\) given \((G, S, X)\)}

\end{minipage}%
\begin{minipage}{0.05\linewidth}
~\end{minipage}%
\begin{minipage}{0.38\linewidth}

\centering{

\tikzstyle{main node}=[circle,draw,font=\sffamily\Large\bfseries]
\tikzstyle{sub node}=[circle,draw,dashed,font=\sffamily\Large\bfseries]
\begin{tikzpicture}[->,>=stealth',shorten >=1pt,auto,node distance=2.5cm,thick]

\node[main node] (G) {$G$};
\node[sub node] (R) [below of=G] {$R$};
\node[main node] (S) [below left=0.5cm and 0.5cm of R] {$S$};
\node[main node] (Y) [right of=G] {$Y$};
\node[main node] (X) [below of=Y] {$X$};

\path[every node/.style={font=\sffamily\small}]
(R) edge node  {} (G)
(R) edge node  {} (X)
(R) edge node  {} (Y)
(R) edge node  {} (S)
(G) edge [dash pattern=on 0.0pt off 20pt, bend left=15] node  {} (X)
(X) edge [dashed, bend right=15] node  {} (G)
(G) edge node  {} (Y)
(X) edge node  {} (Y);

\end{tikzpicture}

}

\subcaption{\label{fig-dag-ys}\textbf{Assumption \ref{a-ci-ys}}:
Conditional independence of \(Y\) and \(S\) given \((G, R, X)\)}

\end{minipage}%
\begin{minipage}{0.10\linewidth}
~\end{minipage}%

\caption{\label{fig-dag}Possible causal structures for which each of the
labeled assumptions is satisfied, represented as a directed acyclic
graph (DAG) where \(G\) is residence location, \(R\) is race, \(S\) is
surname, \(X\) is observed covariates, and \(Y\) is the outcome. Race
\(R\) is unobserved, which is signified by a dashed node boundary. Both
DAGs also satisfy Assumption \ref{a-ci-sg}: Conditional independence of
\(S\) and \((G,X)\) given \(R\).}

\end{figure}%

But in other settings, Assumption \ref{a-ci-yr} may be more plausible.
For example, for a manager reviewing job applications on paper, race is
typically unobserved, but the manager may be influenced by racial or
gender cues in a candidate's name or address
\citep{park2009names, aaslund2012names}. So long as we observe all
information used by the manager and use it in the BISG estimation, the
weighting estimator would give an asymptotically unbiased estimate.
Similarly, in evaluating the fairness of algorithmic decision-making, as
long as all the information used by the algorithm is incorporated into
BISG, Assumption \ref{a-ci-yr} would be appropriate. Outside of these
cases, however, the weighting estimator is likely to be biased.

Finally, while the discussion in this section has been focused on the
BISG methodology, the qualitative results and necessary assumptions
carry over to other approaches which produce probabilistic predictions
of individual race
\citep{zrp, argyle2022misclassification, imai2022addressing, decter2022, greengard2023bisg}.
Just like standard BISG, all of these methods are based on individual
names, geographic location (and sometimes other geographic attributes,
like ACS or decennial census statistics), and possibly additional
individual covariates. Just as with BISG, well-calibrated probabilities
are not generally sufficient to produce unbiased estimates of racial
disparities using the standard weighting or thresholding estimators.

\section{The Proposed Methodology}\label{sec-est}

In this section, we propose an alternative identification strategy for
racial disparities that allow race to directly affect the outcome of
interest. We show that racial disparity is identifiable if surname is
conditionally independent of the outcome given race, residence
geolocation, and other observed information, and the aforementioned
assumptions required by BISG hold. We develop a class of statistical
models, called Bayesian Instrumental Regression for Disparity Estimation
(BIRDiE), that estimate racial disparity under this identification
condition by using surnames as a high-dimensional instrumental variable
for race. These models take as inputs the BISG probabilities, and so can
be easily applied on top of existing analysis pipelines, including those
with alternative probabilistic race prediction methodologies. We also
discuss computational strategies that can be used to apply BIRDiE models
to large data sets, and an extension of the methodology to include an
additional explanatory variable that was not used at the BISG stage.
Finally, we show how to addresses the potential violations of the key
identification assumptions, such as those caused by name-based
discrimination.

\subsection{New Identification Strategy}\label{sec-ident}

To reduce the potential bias of the weighting estimator, we propose an
alternative identification assumption that may be applicable when
Assumption \ref{a-ci-yr} is not credible. Specifically, we assume that
surname, rather than race, satisfies the exclusion restriction
conditional on (unobserved) race, residence location, and other observed
characteristics.

\begin{assump}[Conditional independence of outcome and name]{CI-YS}
For all $i$, $$Y_i\indep S_i\mid R_i,G_i,X_i.$$
\end{assump}

Figure~\ref{fig-dag-ys} shows one possible causal DAG that meets this
assumption as well as Assumption \ref{a-ci-sg}. In this DAG, race can
have a direct effect on the outcome \(Y\) as well as on residence
location \(G\) and other observed characteristics \(X\), while all paths
from \(S\) to \(Y\) are blocked by \(G\), \(X\), or \(R\).

This causal structure is plausible in many real-world settings because,
unlike Assumption \ref{a-ci-yr}, Assumption \ref{a-ci-ys} allows race to
directly affect the outcome. For an outcome like party registration,
Assumption \ref{a-ci-ys} would mean that among White voters in a
particular geographic region, voters named Smith would \emph{a priori}
be no more or less likely to identify with one party than voters named
Thomas. In contrast, Assumption \ref{a-ci-yr} would mean that among
voters named Smith in a particular geographic region, White voters would
be \emph{a priori} be no more or less likely to identify with one party
than Black voters. In this case, Assumption \ref{a-ci-yr} is likely to
be violated, while Assumption \ref{a-ci-ys} is plausible. It is
important to note a key trade-off between the two assumptions. While
Assumption \ref{a-ci-ys} rules out the possibility that surname directly
affects the outcome (e.g., name-based discrimination), such a direct
effect is allowed under Assumption \ref{a-ci-yr}. Section~\ref{sec-sens}
revisits this important issue.

The appropriateness of each assumption depends on a specific
application. In the above hiring example, if the manager reviews
applicants anonymously, there will be no name-based discrimination and
the assumption is likely to be satisfied. The assumption may be violated
in other contexts, however. For example, in studying turnout, if
campaigns use the surnames of individual voters to decide whether to
mobilize them, Assumption \ref{a-ci-ys} will be violated. Yet another
possible violation of the assumption is the existence of unobserved
confounder that affects both outcome and surname. The country of origin
for an immigrant may represent such a confounder: where surnames are
informative of country of origin, even within racial groups (as is often
the case for Asian individuals), variations in outcomes by country of
origin will likely violate \ref{a-ci-ys}. This reflects the limitations
of the relatively coarse racial classifications used in BISG, as
discussed above. Even in these cases, however, conditioning on
(unobserved) race is likely to substantially reduce the magnitude of
association between outcome and surnames. In Section~\ref{sec-sens}, we
show how to address this potential violation of Assumption
\ref{a-ci-ys}.

We briefly note that Assumptions \ref{a-ci-yr} and \ref{a-ci-ys} are not
necessarily mutually exclusive. For example, neither race or surname
could have a direct causal effect on the outcome. In these cases, both
the weighting estimator and our approach proposed below could give
reasonable answers, though the data requirements of each may differ.

The following theorem formally shows that it is possible to
nonparametrically identify racial disparities under Assumption
\ref{a-ci-ys}. The proof of the theorem and all other results in this
paper are deferred to the appendix.

\begin{restatable}[Nonparametric Identification]{theorem}{thmid} \label{thm:id}
For any given $g\in\cG$, $x\in\cX$, and $y\in\cY$, define a matrix $\vb P\in\R^{|\cS|\times|\cR|}$ with entries $p_{sr}=\Pr(R=r\mid G=g, X=x, S=s)$  and a vector $\vb b\in\R^{|\cS|}$ with entries $b_s=\Pr(Y=y\mid G=g, X=x, S=s)$.
Then, under Assumption~\ref{a-ci-ys}, and assuming knowledge of the joint distribution $\Pr(R,G,X,S)$, the conditional probabilities $\Pr(Y=y\mid R, G=g, X=x)$ are identified if and only if both $\vb P$ and the augmented matrix $\mqty(\vb P&\vb b)$ have rank $|\cR|$.
\end{restatable}

The essence of this identification result is the following simple
observation. Under Assumption \ref{a-ci-ys}, we have, for all
\(y\in\cY\), \(g\in\cG\), \(x\in\cX\), and \(s\in\cS\),
\begin{equation} \label{eq:tprob} 
    \Pr(Y\,{=}\,y\mid G\,{=}\,g, X\,{=}\,x, S\,{=}\,s)
    = \sum_{r\in\cR} \Pr(Y\,{=}\,y\mid R\,{=}\,r, G\,{=}\,g, X\,{=}\,x)
    \Pr(R\,{=}\,r\mid G\,{=}\,g, X\,{=}\,x, S\,{=}\,s).
\end{equation} The leftmost term is estimable from the data and
corresponds to the vector \(\vb b\) in Theorem \ref{thm:id}, while the
rightmost term is the BISG estimand and corresponds to the matrix
\(\vb P\) in Theorem \ref{thm:id}. Lastly, the remaining term in the
middle can be solved for, since Equation \eqref{eq:tprob} holds across
all combinations of \(Y\), \(G\), \(X\), and \(S\), leading to a large
system of linear equations. Specifically, we have
\((|\cY|-1)\times|\cG|\times|\cX|\times|\cS|\) equations with
\((|\cY|-1)\times|\cG|\times|\cX|\times|\cR|\) unknowns. Since
\(|\cR| \ll |\cS|\), we can identify these unknowns as long as the
linear system has sufficient rank. Our result is closely related to
causal identification based on proxy variables in the presence of
unmeasured confounding
\citep{kuroki2014measurement, miao2018identifying, knox2022proxy}. Here,
we use surname as a proxy variable for (unobserved) race to identify
racial disparities.

Together with Proposition \ref{p:bisg}, Theorem \ref{thm:id} implies
that racial disparities can be identified under Assumptions
\ref{a-ci-sg}, \ref{a-acc}, and \ref{a-ci-ys}. The identifying equation
\eqref{eq:tprob} shows that \(\Pr(Y=y\mid G=g, X=x, S=s)\) is linear in
the BISG estimands \(\Pr(R=r\mid G=g, X=x, S=s)\). Thus, it is natural
to consider the following least-squares estimator of
\(\Pr(Y=y\mid R, G=g, X=x)\) under this alternative identification
strategy, \[
    \hat{\vb*\mu}^{(\text{ols})}_{Y\mid RGX}(y\mid \cdot, g,x) 
    =  (\hat{\vb P}_{\cI(xg)}^\top \hat{\vb P}_{\cI(xg)})^{-1}\hat{\vb P}_{\cI(xg)}\,\ind\{{\vb Y}_{\cI(xg)} = y\},
\] where as above \(\hat{\vb P}\) is the matrix of BISG probabilities,
and \(\cI(xg)\) is the set of individuals \(i\) with \(X_i=x\) and
\(G_i=g\). Here and throughout the paper, a dot will indicate a vector
constructed over that index, so
\(\hat{\vb*\mu}^{(\text{ols})}_{Y\mid RGX}(y\mid \cdot, g,x)\) is a
vector of conditional probabilities for a particular outcome level \(y\)
across all racial groups in \(\cR\). By post-stratifying this estimator
across the \((G,X)\) cells, we arrive at an estimator of
\(\Pr(Y=y\mid R)\), \[
\hat{\vb*\mu}^{(\text{p-ols})}_{Y\mid R}(y\mid r) 
= \sum_{x\in\cX,g\in\cG} (\hat{\vb P}_{\cI(xg)}^\top \hat{\vb P}_{\cI(xg)})^{-1} 
\hat{\vb P}_{\cI(xg)}\,\ind\{{\vb Y}_{\cI(xg)} = y\})_r q_{gx|r} ,
\] since \(q_{gx|r}=\Pr(G=g,X=x\mid R=r)\) under Assumption \ref{a-acc}.
This estimator is unbiased, as the following theorem shows.

\begin{restatable}[Unbiasedness of OLS Estimator]{theorem}{thmolsunb} \label{thm:ols-unb}
If Assumptions \ref{a-ci-sg}, \ref{a-acc}, and \ref{a-ci-ys} hold, 
and the identification conditions in Theorem~\ref{thm:id} are satisfied,
then for all $y\in\cY$ and $r\in\cR$, \[
    \E[\hat{\mu}^{(\text{p-ols})}_{Y\mid R}(y\mid r)] = \Pr(Y=y\mid R=r).
\]
\end{restatable}

It is worth comparing this OLS estimator with the weighting estimator
\(\hat{\mu}^{(\text{wtd})}_{y\mid r}\). The next theorem shows that
within the \((G,X)\) cells these two estimators are guaranteed to
disagree, unless either the BISG probabilities perfectly discriminate or
the weighting estimator is constant across races. Unfortunately, these
two conditions are almost never met in practice. This underscores the
importance of selecting the appropriate assumption (Assumption
\ref{a-ci-yr} or \ref{a-ci-ys}) for a particular analysis, since they
imply different estimators with different results.

\begin{restatable}[Necessary and Sufficient Condition for Equality of the Weighting and OLS Estimators]{theorem}{thmwtdols} \label{thm:wtd-vs-ols}
For any $y\in\cY$, $g\in\cG$ and $x\in\cX$, within the set of individuals with $G_i=g$ and $X_i=x$, we have that $\hat{\vb*\mu}^{(\text{wtd})}_{Y|R}(y\mid\cdot) = \hat{\vb*\mu}^{(\text{ols})}_{Y|R}(y\mid\cdot)$
if and only if for every pair $j,k\in\cR$, either the BISG probabilities perfectly discriminate 
(i.e., $\Pr(R_i=j\mid G_i,X_i,S_i)>0$ implies $\Pr(R_i=k\mid G_i,X_i,S_i)=0$ and vice versa) or 
$\hat{\mu}^{(\text{wtd})}_{Y|R}(y\mid j)=\hat{\mu}^{(\text{wtd})}_{Y|R}(y\mid k)$.
\end{restatable}

Despite potential advantages over the weighting estimator, the OLS
estimator is not well-suited to estimation in practice, since it ignores
the fact that the unknown parameters are probabilities and thus
constrained to be nonnegative and sum to 1. As a result, in any
particular sample, the estimator can produce impossible or contradictory
estimates. This is particularly problematic because a large number of
unique surnames make both \(\vb P\) and \(\vb b\) high-dimensional. To
address this challenge, and to open the door to more flexible modeling,
we next propose a Bayesian modeling approach that is based on our
identification strategy and yet satisfies necessary constraints.

\subsection{Bayesian Instrumental Regression for Disparity
Estimation}\label{sec-birdie}

The BIRDiE approach combines a user-specified complete-data outcome
model \(\pi(Y\mid R, G, X, \Theta)\), parametrized by \(\Theta\), with
the BISG model in order to estimate the distribution \(Y\mid R\) that is
of interest. In this regard, it mirrors the two-stage instrumental
variables (IV) regression: a first stage (BISG) that estimates the
relationship between instrument (surname) and variable of interest
(race), and a second stage (BIRDiE) that uses the first-stage estimates
to produce valid estimates of the quantity of interest. Unlike two-stage
IV, however, the BIRDiE approach is based on a coherent joint
distribution of data, unknown parameters, and race. We exploit this fact
and develop the general BIRDiE modeling approach below.

Specifically, the BIRDiE posterior is obtained by applying Assumptions
\ref{a-ci-sg}, \ref{a-acc}, and \ref{a-ci-ys} to the joint distribution
\(\pi(Y, R, G, X, S, \Theta)\): \begin{align}
    \pi(\Theta, \vb R\mid \vb Y, \vb G, \vb X, \vb S)
    &\propto \pi(\Theta, \vb R, \vb Y, \vb G, \vb X, \vb S) \nonumber \\
    &\propto \pi(\Theta)\prod_{i=1}^N \pi(Y_i\mid R_i, G_i, X_i, \Theta)
            \pi(R_i\mid G_i, X_i, S_i) \nonumber \\
    &= \pi(\Theta) \prod_{i=1}^N \pi(Y_i\mid R_i, G_i, X_i, \Theta) \hat{P_i}_{R_i}, \label{eq:posterior}
\end{align} where as above \(\hat{\vb P}_i\) are the BISG probability
estimates for individual \(i\), which depend on Census data represented
in \(\vb q_{GX\mid R}\), \(\vb q_{S\mid R}\), and \(\vb q_R\), but not
on the outcome-model parameters \(\Theta\). As a result, these
``first-stage'' BISG estimates can be plugged directly into the BIRDiE
posterior computation without any loss of Bayesian coherency. This
remains true if a more complex model is used in place of the BISG
probabilities \citetext{\citealp[such as those
of][]{zrp}; \citealp{imai2022addressing}; \citealp{argyle2022misclassification}; \citealp[and][]{decter2022}},
so long as the parameters of the first-stage model are \emph{a priori}
independent of the parameters of the BIRDiE model.

To apply this general BIRDiE model to a particular analysis requires
choosing a complete-data outcome model, given by the likelihood
\(\pi(Y_i\mid R_i, G_i, X_i, \Theta)\) and prior \(\pi(\Theta)\). Since
\(Y\) is discrete, a categorical regression model is appropriate for
\(\pi(Y_i\mid R_i, G_i, X_i, \Theta)\). The parametrization of the
categorical regression will depend on the analyst's goals, computational
resources, and prior beliefs about the structure of the problem. We
present here several reasonable alternatives that trade off modeling
flexibility and computational efficiency.

\subparagraph{Complete-pooling model.}\label{complete-pooling-model.}

The simplest possible model is one in which the relationship between
\(Y\) and \(R\) does not vary with \(G\) or \(X\). This model is
parametrized by \(\Theta=\{\vb*\theta_r\}_{r\in\cR}\), which describe
the distribution of \(Y\) within every level of \(R\): \begin{align*}
        Y_i\mid R_i, G_i, X_i, \Theta &\sim \Categorical_{\cY}(\vb*\theta_{R_i}) \\
        \vb*\theta_r &\iid \Dirichlet(\vb*\alpha),
\end{align*} where \(\Categorical_{\cY}\) denotes a discrete
(categorical) distribution on the set \(\cY\). With known \(\vb R\), the
posterior of \(\vb*\theta_r\) is conjugate, a fact which will make
computation under the EM scheme described in Section~\ref{sec-compute}
below extremely efficient. Of course, this efficiency comes at the cost
of a restrictive model that allows for no role of \(G\) and \(X\). If in
reality \(\Pr(Y\mid R)\) does vary along these dimensions, it is
possible that the posterior of \(\vb*\theta_r\) will not accurately
estimate \(\Pr(Y\mid R)\). In any case, if the analyst is interested in
subgroup or small-area estimates of \(\Pr(Y\mid R)\), the
complete-pooling model will be of little use.

\subparagraph{Saturated (no-pooling)
model.}\label{saturated-no-pooling-model.}

At the other end of the spectrum from the complete-pooling model is a
\emph{saturated} or no-pooling model, which estimates a different
distribution of \(Y\mid R\) within every level of \(G\) and \(X\):
\begin{align*}
        Y_i\mid R_i, G_i, X_i, \Theta &\sim \Categorical_{\cY}(\vb*\theta_{R_iG_iX_i}) \\
        \vb*\theta_{rgx} &\iid \Dirichlet(\vb*\alpha).
\end{align*} This model is closest to the OLS estimator, though it
ensures that all probability estimates lie in \([0, 1]\). As with the
complete-pooling model, the posterior of \(\vb*\theta_{rgx}\) is
conjugate to its prior, and so computation can be made efficient.
Additionally, this model allows for any arbitrary relationship between
\(Y\), \(R\), \(G\), and \(X\). Since it is fully nonparametric, the
posterior will converge to the true \(\Pr(Y\mid R, G, X)\) with enough
data in each \((G, X)\) cell. However, in practice, the model can suffer
from the curse of dimensionality: the number of \((G, X)\) cells may be
relatively large compared to the amount of available data, or even
exceed it, especially since \(G\) can be quite large, covering many
blocks or ZIP codes. In these cases, the prior will dominate the data in
each cell, which could have a large biasing effect even on overall
inferences about \(\Pr(Y\mid R)\).

\subparagraph{General mixed-effects
model.}\label{general-mixed-effects-model.}

As a compromise between the complete-pooling and no-pooling model, a
partial pooling approach based on a multinomial mixed-effects model can
be used. Properly specified, the mixed-effects model maintains the
flexibility of the saturated model while avoiding its high bias and
variance in finite samples. \begin{align*}
    Y_i\mid R_i, G_i, X_i, \Theta &\sim \Categorical_{\cY}(g^{-1}(\vb*\mu_{rgx})) \\
    \mu_{rgxy} &= \vb W\vb*\beta_{ry} + \vb Z\vb u_{ry} \\
    \quad \vb u_{ry}\mid \phi_{ry} &\sim \Norm(0, \Sigma(\vb*\phi_{ry})) \\
    \vb*\beta_{ry} &\iid f^{(\beta)}_r, \quad \vb*\phi_{ry} \iid f^{(\phi)}_r ,
\end{align*} where \(g^{-1}\) is a softmax or other link function,
\(\vb W\) and \(\vb Z\) are matrices of fixed and random effects,
respectively, \(\vb*\phi\) is a vector of random-effect parameters, and
\(f^{(\beta)}\) and \(f^{(\phi)}\) are some priors for the
super-scripted parameters. Some fixed or random effects could be shared
across combinations of \(R\) and \(Y\), though this could complicate
computation. In practice, we recommend including \(X\) and especially
\(G\) in the model as random effects, with hierarchical structure as
appropriate. Such a structure partially pools estimates of
\(\Pr(Y\mid R, X, G)\) towards an overall estimate of \(\Pr(Y\mid R)\),
allowing the model to share information between geographic areas. This
should prove especially useful in cases where some areas have few or no
observations for certain racial groups.

We also recommend including group-level covariates as fixed effects,
which will help share information across random effects and
significantly improve generalization performance to unseen random effect
levels \citep{buttice2013does}. For example, if \(G\) records counties,
analysts could include racial and socioeconomic variables measured at
the county level as predictors. This would help produce more accurate
estimates of \(\Pr(Y\mid R, G)\) to the extent that variation in these
probabilities is associated with these racial and socioeconomic
variables. Ultimately the structure of this general model will have to
be chosen based on the data and the relevant research question.

\subsection{Computation}\label{sec-compute}

The posterior in Equation \eqref{eq:posterior} contains the
high-dimensional discrete nuisance parameter \(\vb R\), which poses a
challenge for computation. We suggest two approaches for handling
\(\vb R\), one suited to small sample sizes, and one suited to large
sample sizes. We also discuss uncertainty quantification for conjugate
complete-data models.

\subparagraph{Small samples: Inference directly on the marginal
likelihood.}\label{small-samples-inference-directly-on-the-marginal-likelihood.}

Since \(\vb R\) is discrete, we can marginalize it out as follows: \[
    \pi(\Theta\mid \vb Y, \vb G, \vb X, \vb S)
    = \sum_{\vb r\in \cR^N} \pi(\Theta, \vb r\mid \vb Y, \vb G, \vb X, \vb S)
    \propto \pi(\Theta) \prod_{i=1}^N \sum_{r\in\cR} \pi(Y_i\mid r, G_i, X_i, \Theta)\hat{P}_{ir}.
    \numberthis \label{eq:post-marg}
\] This decouples the total number of parameters from the sample size.
Additionally, Equation \eqref{eq:post-marg} has only continuous
parameters, and so can be used with any general Bayesian inference
procedure such as Markov chain Monte Carlo (MCMC). However, the
moderate-to-high dimensionality in practical settings, even after
integrating out \(\vb R\), and the expensive likelihood calculation due
to the sum nested within the outer product, makes MCMC algorithms
computationally too expensive outside of relatively small datasets.

\subparagraph{Large samples:
Expectation-Maximization.}\label{large-samples-expectation-maximization.}

When the number of individuals exceeds a thousand or so, we propose an
Expectation-Maximization (EM) algorithm \citep{dempster1977} to
calculate the maximum \emph{a posteriori} (MAP) estimate of \(\Theta\)
for BIRDiE models. The EM algorithm alternates between an E-step which
calculates the expected log posterior density \(Q\), averaging over the
missing \(\vb R\), and an M-step which maximizes \(Q\) over values of
\(\Theta\).

Specifically, given a current parameter estimate \(\Theta^{(t)}\), the
expected log posterior density can be written as \begin{align*}
    Q(\Theta^{(t+1)}\mid \Theta^{(t)})
    &= \E_t[\log \pi(\Theta^{(t+1)}, \vb R \mid \vb Y, \vb G, \vb X, \vb S)] \\
    &= \sum_{\vb r\in\cR^N} 
        \pi(\vb r \mid \Theta^{(t)}, \vb Y, \vb G, \vb X, \vb S)
        \log \pi(\Theta^{(t+1)}, \vb r \mid \vb Y, \vb G, \vb X, \vb S) \\
    &= \log \pi(\Theta^{(t+1)}) + \sum_{i=1}^N\sum_{r\in\cR} \Big\{
        \qty(\log \pi(Y_i\mid r, G_i,, X_i, \Theta^{(t+1)}) + \log\hat P_{ir}) \\
        &\qquad\qquad\qquad\qquad\qquad\qquad\times 
        \pi(R_i=r \mid \Theta^{(t)}, Y_i, G_i, X_i, S_i) \Big\} \\
    &= C + \log \pi(\Theta^{(t+1)}) + \sum_{i=1}^N\sum_{r\in\cR} 
        \tilde{P}_{ir\mid Y}^{(t)} \log \pi(Y_i\mid r, G_i, X_i, \Theta^{(t+1)}),
\end{align*} where \(C\) is a constant that can be ignored as it does
not depend on \(\Theta^{(t+1)}\), and
\(\tilde{\vb P}_{\mid Y}^{(t)}=\pi(\vb R \mid \Theta^{(t)}, \vb Y, \vb G, \vb X, \vb S)\)
are the BISG probabilities updated with Bayes' rule using the outcome
\(\vb Y\): \begin{equation}
    \tilde{P}_{ir\mid Y}^{(t)}
    = \frac{\pi(Y_i\mid r, G_i,, X_i, \Theta^{(t)}) \hat P_{ir}}{
        \sum_{r'\in\cR} \pi(Y_i\mid r', G_i,, X_i, \Theta^{(t)}) \hat P_{ir'} }. \label{eq:update-bisg}
\end{equation} At the M-step, \(Q(\Theta^{(t+1)}\mid \Theta^{(t)})\) is
straightforward to maximize, since it is just the log complete-data
posterior, with likelihood weights given by the
\(\tilde{P}_{ir\mid Y}^{(t)}\). Additionally, if \(\Theta\) can be
partitioned into parameters which only affect individuals in each racial
group (as is the case with all the models described in
Section~\ref{sec-birdie} above), then the maximization can be performed
separately on each group of individuals.

A critical advantage of this EM scheme over working directly with the
marginal likelihood is that the maximization in the M-step can be
performed using sufficient statistics calculated as part of the E-step,
rather than on all of the individual entries in the data. Since the
M-step is usually the practical (if not also asymptotic) bottleneck in
the computation, this is enormously helpful---the problem size scales
with \(|\cY|\times|\cX|\times|\cG|\) rather than with \(N\).
Specifically, notice that we can rewrite
\(Q(\Theta^{(t+1)}\mid \Theta^{(t)})\) (dropping the unnecessary
constant) as \begin{align*}
    Q(\Theta^{(t+1)}\mid \Theta^{(t)})
    &= \log \pi(\Theta^{(t+1)}) + \sum_{i=1}^N\sum_{r\in\cR} 
        \tilde{P}_{ir\mid Y}^{(t)} \log \pi(Y_i\mid r, G_i, X_i, \Theta^{(t+1)}) \\
    &= \log \pi(\Theta^{(t+1)}) + \sum_{r\in\cR}\sum_{y\in\cY}\sum_{x\in\cX}\sum_{g\in\cG} 
        \log \pi(y\mid r, g, x, \Theta^{(t+1)}) 
        \qty(\sum_{i\in\cI(yxg)} \tilde{P}_{ir\mid Y}^{(t)}),
\end{align*} where analogously to above \(\cI(yxg)\) is the set of
individuals with \(Y_i=y\), \(X_i=x\), and \(G_i=g\).

For BIRDiE models where the complete-data likelihood is conjugate to the
prior, such as the complete- and no-pooling models, these sufficient
statistics are used in the M-step anyway, and can be efficiently
calculated during the E-step as well. In combination with the
acceleration scheme described next, this allows the entire EM algorithm
to be run to convergence on data with hundreds of thousands or millions
of individuals in a matter of seconds.

While EM algorithms are highly stable, due to their monotonic increasing
of the marginal likelihood, they are also often slow to converge
\citep{laird1993}. To address this, we propose, and include in our
open-source software implementation, the use of fixed-point iteration
accelerators such as Anderson acceleration or SQUAREM
\citep{varadhan2008simple}. These techniques can substantially reduce
the overall computational time without meaningfully affecting the
stability of inference.

\subparagraph{Uncertainty quantification via blocked Gibbs
sampling.}\label{uncertainty-quantification-via-blocked-gibbs-sampling.}

While computationally efficient, the EM algorithm does not provide any
uncertainty quantification. However, in large samples, sampling and
model-based uncertainty dominated by biases caused by even small
violations of the underlying assumptions, a problem we discuss below in
Section~\ref{sec-sens} and the accompanying appendix. Even so, it is
often useful to have some measure of sampling uncertainty. This is
possible using blocked Gibbs sampling for some BIRDiE models with
conjugate complete-data posteriors, such as the complete- and no-pooling
models described above.

Our Gibbs sampling strategy is to alternate sampling from the
model-updated BISG probabilities
\(\pi(\vb R\mid \vb Y, \vb G, \vb X, \vb S, \Theta)\) and sampling from
the complete-data posterior
\(\pi(\Theta\mid \vb Y, \vb R, \vb G, \vb X, \vb S)\). The first step
involves the same calculations as the E-step in Equation
\eqref{eq:update-bisg}. The second step is computationally tractable in
medium-to-large samples when the complete-data likelihood is conjugate
to the prior, as is the case for the Categorical-Dirichlet pooling
models proposed above.

Categorical-Dirichlet models with latent discrete variables, like the
pooling BIRDiE models, are often tackled using a collapsed Gibbs sampler
where \(\Theta\) has been marginalized out. We found that approach to be
unsuccessful here, however, since the one-by-one updating of each
individual's \(R_i\) meant that the sampler was unable to traverse to
the correct region of parameter space and got stuck near the BISG
initialization. In contrast, the blocked Gibbs sampler is vectorized
over individuals and rapidly moves to the mode identified by the EM
algorithm.

For non-conjugate BIRDiE models with few parameters, bootstrapping the
EM procedure is computationally feasible and can be used to approximate
the covariance matrix of the MAP estimate.

\subsection{Updated Individual Race
Probabilities}\label{sec-update-bisg}

The EM algorithm produces the updated individual race probabilities
given in Equation \eqref{eq:update-bisg}. One feature of these updated
probabilities is that it is appropriate to apply the weighting estimator
to them to estimate disparities. This is because the updated
probabilities condition on \(\vb Y\), and so the asymptotic bias term in
Theorem \ref{thm:wt-bias} becomes zero. In fact, the weighting estimate
from the updated probabilities is identical to the BIRDiE estimate.
While more study is required, for downstream settings where weights are
needed, generating these weights with BISG followed by BIRDiE will
likely produce more accurate results than simply using BISG weights
alone.

\subsection{Additional Explanatory Variables}\label{sec-addlcov}

Often, researchers are interested in not just \(\Pr(Y\mid R)\) but also
\(\Pr(Y\mid W, R)\), for some variable \(W\in\cW\) which is not part of
the BISG predictors \((X, G)\). For example, a lending firm auditing
potential racial disparities in lending decisions would likely be
interested both in how the rate of loan approval (\(Y\)) varies by race,
but also how loan approval varies by race, conditional on a measure of
creditworthiness (\(W\)). The unconditional disparities reflect
realities of systemic racism and inequality, while the conditional
disparities measure the fairness of the firm's lending decisions after
controlling for these systemic factors. Such estimates could be used to
compute various measures of algorithmic fairness, including calibration
parity and false positive error rate balance. Another scenario is a
policy evaluation study, where researchers are interested in how the
impact of policy varies across racial groups. Such an analysis requires
incorporating an interaction between race and the treatment variable.

There are two main ways to perform such an analysis with our proposed
methodology. The first, and perhaps simplest, is to apply the
methodology to the combined variable \(\underline{YW}\in\cY\times\cW\).
This will produce estimates of \(\Pr(Y, W\mid R)\), from which
\(\Pr(Y\mid W, R)\) can be straightforwardly calculated by appropriate
normalization. This approach will work well if \(|\cY|\) and \(|\cW|\)
are both small, so that \(|\cY\times\cW|\) is of manageable size. If one
of these variables has many levels, however, directly estimating the
distribution of \(Y, W\mid R\) could be less efficient, as it does not
account for information about the marginal distributions \(Y\mid R\) and
\(W\mid R\).

An alternative approach is to first apply the proposed methodology to
estimate \(\Pr(W\mid R)\). This allows for calculation of model-updated
BISG probabilities
\(\tilde{\vb P}_{\mid W}=\pi(\vb R \mid \hat\Theta, \vb W, \vb G, \vb X, \vb S)\),
which are also computed as a byproduct of the EM algorithm described
above. Then, the methodology can be applied again, using
\(\tilde{\vb P}_{\mid W}\) as the input probabilities rather than the
original BISG probabilities, to estimate \(\Pr(Y\mid W, R)\). This
approach will likely perform better when \(W\) consists of multiple
predictors or if either \(|\cY|\) or \(|\cW|\) are large.

Both of these approaches require the following identifying assumption,
which generalizes Assumption \ref{a-ci-ys}.

\begin{assump}[Conditional independence of outcome, predictor and name]{CI-YWS}
For all $i$, $(Y_i, W_i)\indep S_i\mid R_i,G_i,X_i$, or, equivalently, \[
    W_i\indep S_i\mid R_i,G_i,X_i \qand
    Y_i\indep S_i\mid W_i,R_i,G_i,X_i. 
\]
\end{assump}

In the lending example, this would translate to the assumption that a
measure of creditworthiness is independent of last name after
controlling for race, location, and covariates, and that lending
decisions are independent of last names after controlling for
creditworthiness, race, location, and covariates.

\subsection{Addressing Potential Violations of the
Assumptions}\label{sec-sens}

BIRDiE crucially relies on Assumption \ref{a-ci-ys} for identification.
In addition, like the weighting and thresholding estimators, it also
relies upon Assumptions \ref{a-ci-sg} and \ref{a-acc}, which are
required for the BISG race probabilities to be accurate. Unfortunately,
these assumptions may not exactly hold in practice, and are also not
testable in observed data. In this section and Appendix
\ref{sec-app-sens}, we develop sensitivity analyses that assess how
violations of these assumptions affect the estimates of racial
disparities.

First, BIRDiE assumes that conditional on unobserved race and observed
covariates, outcomes and surnames are independent. As discussed in
Section~\ref{sec-ident}, however, association between the outcome and
country of origin or racial subgroups may lead to correlation between
surnames and outcome even after controlling for race and geography. To
address this problem, suppose that a low-dimensional summary statistic
of surname, \(f:\cS\to\R^d\), \(d\ll|\cS|\), is available, where \(f\)
may map each surname to a finer ethnic group within each racial
category. For example, Imai is a Japanese name whereas McCartan is a
name of Irish origin. If \(f\) can classify surnames into finer racial
subgroups or countries of origin---even approximately---then it can be
used to control for this channel of possible violations of Assumption
\ref{a-ci-ys}. Formally, we relax Assumption \ref{a-ci-ys} as follows.

\begin{assump}[Partial conditional independence of outcome and name]{CI-YSF}
For all $i$, $$Y_i\indep S_i\mid f(S_i),R_i,G_i,X_i.$$
\end{assump}

The next theorem shows that it is still possible to nonparametrically
identify racial disparities under Assumption \ref{a-ci-ysf} under the
identification condition, which is only slightly stronger than for
Theorem \ref{thm:id}.

\begin{restatable}[Nonparametric Identification Under Assumption \ref{a-ci-ysf}]{theorem}{thmidrel} \label{thm:id2}
Let $f:\cS\to\R^d$, $d< |\cS|$, with range $f(\cS)$.
For any given $g\in\cG$, $x\in\cX$, $z\in f(\cS)$, and $y\in\cY$, define a matrix $\vb P\in\R^{|\cS|\times|\cR|}$ with entries $p_{sr}=\Pr(R=r\mid G=g, X=x, S=s)$  and a vector $\vb b\in\R^{|\cS|}$ with entries $b_s=\Pr(Y=y\mid G=g, X=x, S=s)$.
Then under Assumption~\ref{a-ci-ysf}, and assuming knowledge of the joint distribution $\Pr(R,G,X,S)$, the conditional probabilities $\Pr(Y=y\mid R,f(S)=z,G=g, X=x)$ are identified if and only if both $\vb P$ and the augmented matrix $\mqty(\vb P&\vb b)$ have rank $|\cR|$.
\end{restatable}

As long as the dimension \(d\) of the surname summary statistic \(f(S)\)
is much smaller than the (usually large) number of surnames \(|\cS|\),
racial disparities are likely to be identified under Theorem
\ref{thm:id2} when they are already identified under Theorem
\ref{thm:id}. Thus, Assumption \ref{a-ci-ysf} and Theorem \ref{thm:id2}
can be used in conjunction with carefully chosen \(f\) in order to probe
likely failure modes of the more restrictive Assumption \ref{a-ci-ys}.
If estimates are not much affected by the inclusion of \(f(S)\), then
researchers can be more confident in the plausibility of Assumption
\ref{a-ci-ys}.

Second, bias can also arise from violations of the assumptions
underlying the BISG methodology (Assumptions \ref{a-ci-sg} and
\ref{a-acc}). Of course, this is not unique to the proposed methodology:
violations of these assumptions will also affect the validity of other
disparity estimators such as weighting or thresholding. However, since
as discussed above the BISG assumptions may rarely hold exactly in
practice, we provide in Appendix \ref{sec-app-sens} several results
characterizing how the model's estimates are affected by bias in the
BISG probabilities.

\section{Empirical Validation with the Voter File}\label{sec-valid}

To better understand how BIRDiE performs in real-world contexts, we
apply it to North Carolina voter registration data. This voter file
contains individual-level self-reported race for almost all voters and
hence the ``ground truth'' relationship between outcome and race is
known. We compare the performance of BIRDiE models against those of the
weighting and thresholding estimators. We also evaluate how the
estimation error depends on the level of geographic precision used in
the BISG probabilities. Finally, we briefly demonstrate various
extensions of the BIRDiE methodology: small-area estimates (as discussed
in Section~\ref{sec-birdie}), improved individual race predictions
(Section~\ref{sec-update-bisg}), estimation conditional on an additional
explanatory variable (Section~\ref{sec-addlcov}), and sensitivity
analysis for potential assumption violation (Section~\ref{sec-sens}).

\subsection{North Carolina Voter File}\label{north-carolina-voter-file}

Like most other Southern states, which have a history of
disenfranchising minority voters, the state of North Carolina asks (and
previously required) every voter to self-report their race upon
registration. This data, along with voters' names, addresses, gender,
party registration (if any), and voting history, is part of the voter
file that the secretary of state makes publicly available. This feature
makes the voter file an ideal validation setting for the proposed
methodology. The outcome we examine here, party registration, is the
product of many unobservable factors, and is known to differ across
racial groups. Since self-reported race is available, inferences about
these racial disparities using the estimators discussed here can be
compared to the corresponding ground truth.

Estimation of party registration by race is of substantive interest as
well, especially in the context of the Voting Rights Act of 1965 (VRA).
The relationship between these variables is critical for understanding
the impact of policy changes such as redistricting or election rules on
compliance with the VRA, and for establishing legal standing to
challenge these policies under the VRA. As many states do not ask for
self-reported race during voter registration, methods like BIRDiE are
important tools for evaluating VRA compliance.

We use a subset of the October 2022 voter file which could be linked to
a proprietary voter file provided by L2, Inc., a leading national
non-partisan firm and the oldest organization in the United States that
supplies voter data and related technology to candidates, political
parties, pollsters, and consultants for use in campaigns. The L2 file
geocoded each address to a Census block, which allows for the finest
block-level BISG predictions. We also removed any records without
individual race information, since our goal is validation compared to
some ground truth, rather than inference about the entire population of
registered North Carolina voters. Altogether, 22.1\% of records either
had missing race information or could not be linked to the L2 file.

\begin{figure}

\centering{

\includegraphics[width=6in,height=\textheight]{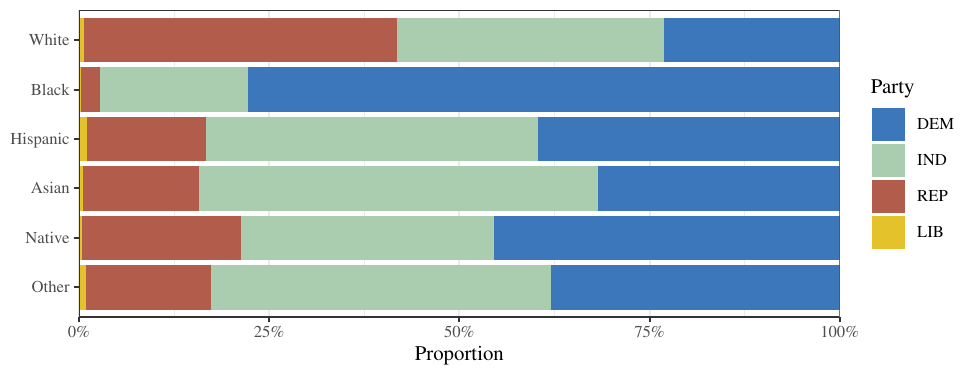}

}

\caption[Distribution of party registration by race for North Carolina
voters]{\label{fig-nc-overview}Distribution of party registration by
race for a sample of 1,000,000 North Carolina voters. Parties are
Libertarian (LIB), Republican (REP), Independent (IND), and Democratic
(DEM).}

\end{figure}%

The overall merged voter file contains 5,754,912 voters, 71.1\% of which
are White, 21.1\% of which are Black, and 7.8\% of which belong to
another race. To reduce computational burden, we further subsampled this
file by selecting 1,000,000 records at random without replacement. This
sample size is large enough to ensure that sampling error in the
estimates is negligibly small.

Figure~\ref{fig-nc-overview} shows the distribution of party
registration by self-reported race in this subsample. White voters
disproportionately register as Republicans, while Black voters
disproportionately register as Democrats. This serves as the
ground-truth in our validation analysis presented below.

\subsection{The Model Setup}\label{the-model-setup}

We first calculate BISG probabilities using 2010 Census data at the
census block, tract, ZIP code tabulation area (ZCTA), and county level.
Every record in the voter file contains county information, while
roughly 13\% of records are missing ZIP codes and 27\% of records are
missing blocks/tracts; when these finer geographic identifiers were
missing, we used county-level Census tables in the BISG calculations.

The BISG probabilities are broadly accurate. Using the maximum a
posteriori racial category as a prediction, we obtain accuracy of 76.2\%
for the county probabilities, 78.4\% for the ZIP code probabilities,
78.5\% for the tract probabilities, and 79.6\% for the block
probabilities. An alternative measure of the quality of the BISG
probabilities is the logarithmic score, a proper scoring rule which
rewards precise and calibrated probabilistic estimates (higher values
are better). The logarithmic scores for the BISG probabilities are
--0.618 for counties, --0.587 for ZIP codes, --0.585 for tracts, and
--0.607 for blocks. For comparison, the prior-only logarithmic score
(i.e., using no name or geographic information) is --0.867. The worsened
performance for block-level versus tract-level probabilities likely
stems from the larger impact of census measurement error at smaller
geographies, a problem that could be addressed using newer BISG methods
such as those of \citet{imai2022addressing}.

Since the goal of our validation study is to compare BIRDiE estimates
with weighting and thresholding estimates, we do not make additional
comparisons between BISG probabilities and those generated with
alternative racial prediction methods. To the extent competing racial
prediction methods improve prediction accuracy, we expect the gap
between different disparity estimation methods (weighting, thresholding,
BIRDiE) to narrow, consistent with Theorems \ref{thm:wt-bias} and
\ref{thm:wtd-vs-ols}. As we have discussed, however, high accuracy of
racial prediction alone is neither necessary nor sufficient for accurate
estimation of racial disparities. If other racial prediction methods
produce increased accuracy at the cost of worsened calibration, accuracy
in estimating racial disparities may be poor whether using weighting,
thresholding, or BIRDiE estimates.

In our validation, for a given set of BISG probabilities, we estimate
the conditional distribution of each outcome variable given race using
BIRDiE with both saturated pooling and multinomial mixed-effects models
introduced above. We then compare the resulting estimates based on these
BIRDiE models against those of the two existing estimators---the
weighting estimator as well as a thresholding estimator that
deterministically assigns each individual the maximum \emph{a
posteriori} racial category. For the saturated BIRDiE model, we use
geographic effects matching the geographic level used in the BISG
probabilities (e.g., county effects for the county-level BISG
probabilities), except for the block-level probabilities. Due to the
large number of individual census blocks, we use tract-level effects
instead for this particular model. For the mixed-effects BIRDiE model,
in addition to these geographic effects we add two geography-level
covariates: the White and Black fraction of the population in each
individual's geography of residence. These covariates should help
further regularize and share information among the individual geographic
effects.

We use noninformative or weakly informative priors for both BIRDiE
models. For the saturated model, we use a uniform prior for the
Dirichlet hyperparameters. For the mixed model, the prior on the fixed
effects \(\vb*\beta_r\) for each racial group is a weakly informative
Normal with standard deviation \(2p_r\), where \(p_r\) is the share of
the racial group in the sample; this encourages more shrinkage for
groups where less data is available. The overall global intercept
received a \(\Norm(0, 5^2)\) prior. The prior on the random intercept
scale is \(\text{Inv-Gamma}(4, 1.5)\), designed to support a range of
possible heterogeneity of the outcome-race relationship across
geographic levels while discouraging the M step from finding a mode at
zero; it places 95\% of its mass between 0.17 and 1.376. The random
intercept correlation matrix (across levels of \(Y\)) received an LKJ
prior with shape parameter 2. Variations of these prior choices did not
noticeably affect the top-line estimates, however, due to the large
quantity of data overall. To give an idea of the computational
efficiency of the proposed method, the maximum runtime of the saturated
of the BIRDiE models fit to estimate party registration was 6.9 seconds
(on a modern laptop with 8GB RAM), and the maximum runtime for the mixed
model estimates was 22.6 seconds.

\subsection{Estimates of Racial Disparity in Party
Registration}\label{estimates-of-racial-disparity-in-party-registration}

\begin{figure}

\centering{

\includegraphics[width=1\textwidth,height=\textheight]{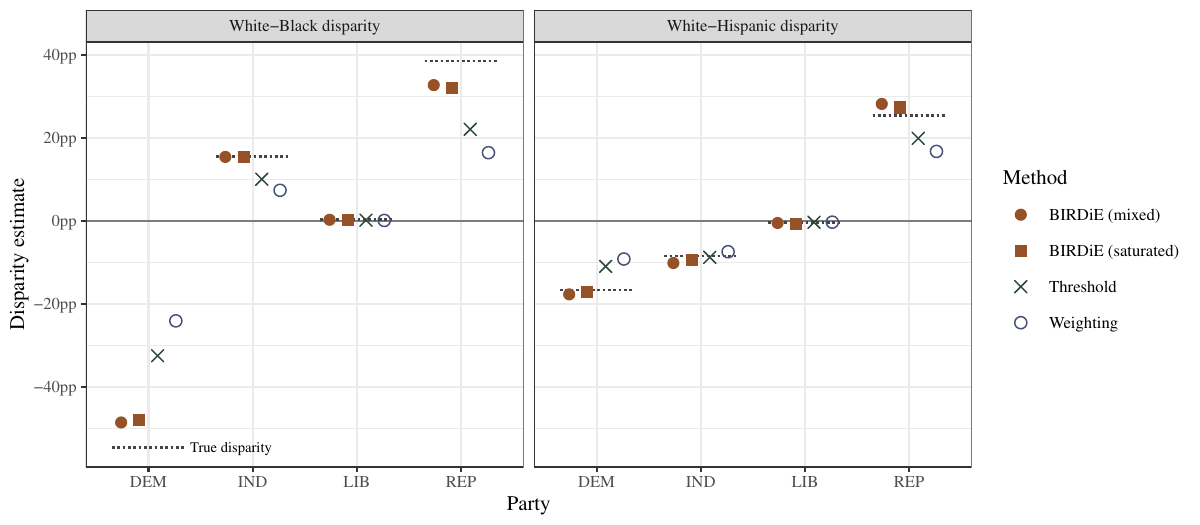}

}

\caption[Error in racial disparity estimates for party registration, by
estimation method]{\label{fig-nc-disp}Error in the White-Black and
White-Hispanic disparity estimates for party registration, by estimation
method. The true disparities are indicated by the dotted lines. All
methods used block-level data for this figure; the results for other
levels of geographic detail are generally similar. Estimation
uncertainty is minimal and hence suppressed from the figure for
clarity.}

\end{figure}%

We first examine the relative accuracy of the proposed methods in
estimating the disparity between White and Black, and White and Hispanic
voters, in party registration. For example, the true difference in
Democratic registration between Black and White voters in the sample is
\(54.6%
\) percentage points, meaning Black voters register Democratic at a much
higher rate. However, the standard weighting approach produces an
estimate of only \(24.1%
\) percentage points for this disparity---less than half the true value.
This is consistent with Corollary \ref{cor:under}, which states that the
weighting estimator tends to underestimate the magnitude of racial
disparity. The thresholding estimator, while slightly better, also
misses the mark, with an estimate of \(32.5%
\) percentage points. In contrast, the saturated BIRDiE model produces
an estimate of \(47.9%
\) percentage points, and the mixed BIRDiE model estimates \(48.5%
\). These estimates are only slightly lower than the ground truth.

Figure~\ref{fig-nc-disp} compares the empirical performance of the
BIRDiE models against that of the weighting and thresholding estimators
across all of these possible disparity measurements, using the
block-level BISG predictions. For White--Black (left plot) and
White-Hispanic (right plot) disparities in party registration, the
BIRDiE models (solid circles and squares) substantially outperform the
two commonly used estimators (open circles and crosses). For two major
parties, both the weighting and thresholding estimators exhibit a
substantial amount of estimation error, for example, exceeding 20
percentage points for the White-Black disparity for the Democratic
party. In contrast, the two BIRDiE models yield a much smaller
estimation error that ranges within several percentage points for all
racial disparity estimates. The saturated and mixed-effects BIRDiE
models perform similarly with no discernible difference.

\begin{figure}

\centering{

\includegraphics[width=1\textwidth,height=\textheight]{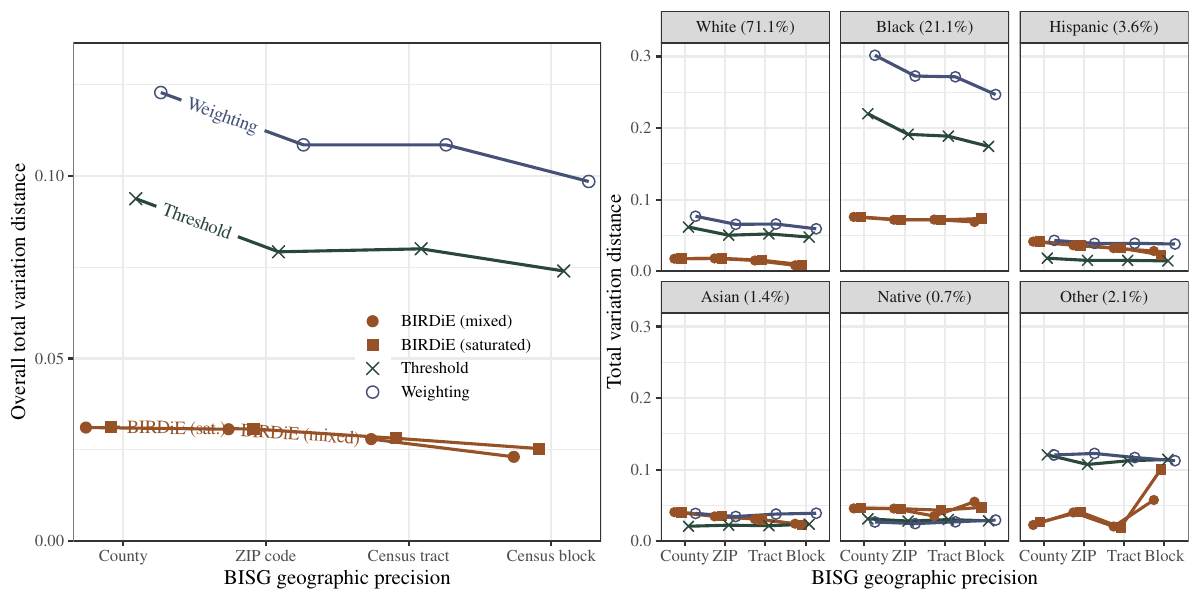}

}

\caption[Total variation distance between the estimated and actual
distribution of party registration, by estimation method and level of
geographic detail used in the BISG predictions]{\label{fig-nc-tv}Total
variation distance between the estimated and actual distribution of
party registration, by estimation method and level of geographic detail
used in the BISG predictions. The left plot shows the overall total
variation distance, whlie the right plot decomposes it by racial group.}

\end{figure}%

For a more comprehensive look at the error in the estimated
partisanship-by-race distributions, we turn to the total variation (TV)
distance, which is calculated as \[
d_\text{TV}(\hat{\vb*\mu}_{Y|R}, \vb*\mu_{Y|R}) 
= \half\sum_{y\in\cY}\sum_{r\in\cR} |\hat{\mu}_{Y,R}(y, r) - \mu_{Y,R}(y, r)|,
\] where \(\mu_{Y,R}\) is the joint distribution of \(Y\) and \(R\). The
TV distance is an upper bound on the error in \emph{any} probability
calculated from the estimated joint distribution, and as such is useful
general-purpose measure of estimation error. The left plot of
Figure~\ref{fig-nc-tv} shows the TV distance for each estimator, not
just for the block-level BISG estimates used in Figure~\ref{fig-nc-disp}
but also across the range of geographic levels used in the BISG
calculation (x-axis).

We find that both BIRDiE models substantially outperform every
alternative method at every geographic level. In general, the estimates
based on the BIRDiE models exhibit a total variation distance whose
magnitude is about one third and one fourth of that for the thresholding
and weighting estimators, respectively. As before, the saturated and
mixed-effects BIRDiE models perform similarly. According to this
measure, we find that finer geographic data provide only minor
improvements in accuracy for the BIRDiE or conventional estimates. While
possibly counterintuitive, this finding underscores the fact that
calibrated BISG probabilities, rather than highly precise probabilities,
are all that is needed for accurate disparity estimation.

Of course, both calibrated and precise probabilities are to be preferred
to imprecise but calibrated probabilities. In practice, however, there
may be a tradeoff between the two. For example, including first names in
the BISG predictions may increase their precision. But, first names may
lead to worse calibration, since BISG methods which use surnames make an
somewhat unrealistic conditional independence assumption, and data on
first names by race come from non-census sources
\citep{tzioumis2018demographic, rosenman2023firstname}. Additionally,
unlike surnames, first names (which are usually chosen by parents) can
be more correlated with socioeconomic variables, leading to violations
of Assumptions \ref{a-ci-ys}.

Lastly, we measure the TV distance for each conditional distribution by
race: \[
d_\text{TV}^{(r)}(\hat{\vb*\mu}_{Y|r}, \vb*\mu_{Y|r}) 
= \half\sum_{y\in\cY} |\hat{\mu}_{Y,r}(y, r) - \mu_{Y,r}(y, r)|.
\] The right plot of Figure~\ref{fig-nc-tv} shows these within-race TV
distances, to illuminate how the estimators perform on each subgroup. In
general, the BIRDiE models are more accurate than the weighting and
thresholding estimators for the White and Black racial groups which
together make up 96\% of the sample. All estimators perform roughly
equally well for Hispanic, Asian, and Native voters (though the
thresholding estimator performs well for Hispanic voters), exhibiting
relatively small estimation error. The weighting and thresholding
estimators perform particularly poorly for Black voters and for the
``Other'' racial group.

\subsection{Small-area Estimation}\label{sec-small}

An advantage of the BIRDiE methodology is its explicit modeling of
\(\Pr(Y\mid R,G,X)\), which produces not only estimates of the marginal
\(\Pr(Y\mid R)\) but also subgroup estimates of how these conditional
distributions vary across covariates and geographic areas. This section
examines the accuracy of the saturated and mixed-effects BIRDiE models
in recovering small-area relationships between party registration and
race, compared to standard methodology that simply applies the weighting
and thresholding estimators within each geographic area. We study
accuracy at the county, ZIP code, and tract level, using the BISG
probabilities and BIRDiE models that were applied to each. Since in
fitting the BIRDiE model to block-level BISG probabilities we used
tract-level random intercepts, we do not present block-level estimates.

\begin{figure}

\centering{

\includegraphics[width=5.5in,height=\textheight]{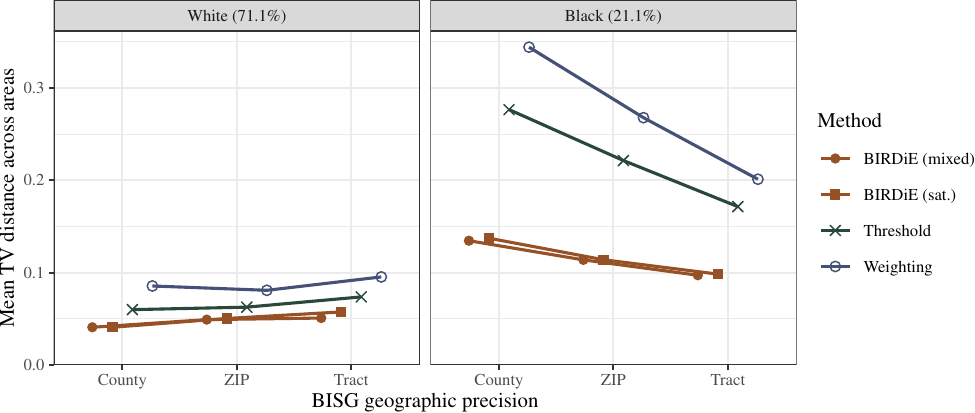}

}

\caption[Accuracy of small-area party registration estimates by
race]{\label{fig-nc-smallarea}Accuracy of small-area estimates by race,
as measured by the average total variation distance.}

\end{figure}%

We evaluate the small-area estimates by calculating the mean total
variation distance between the estimated and true conditional
distributions of party registration by race (averaging across geographic
areas). Figure~\ref{fig-nc-smallarea} summarizes our results, which
qualitatively track the patterns found overall in
Figure~\ref{fig-nc-tv}. The two BIRDiE models exhibit substantially
lower error than the weighting and thresholding estimators. Across all
methods, the error is lower for White voters, who make up the bulk of
the sample. Somewhat surprisingly, the amount of error does not appear
to vary much for the BIRDiE models across different levels of
geography---tract-level estimates are roughly as accurate as
county-level estimates, on average.

Between the BIRDiE models, the mixed-effects model slightly outperforms
the saturated model. This reflects the value in partially pooling
estimates through the random effect structure.

\subsection{Improved Individual Race
Probabilities}\label{improved-individual-race-probabilities}

As discussed in Section~\ref{sec-update-bisg}, we can use the
conditional distribution \(\Pr(Y\mid R, X, G)\) estimated with a BIRDiE
model to create model-updated BISG probabilities
\(\tilde{\vb P}_{\mid Y}=\pi(\vb R \mid \hat\Theta, \vb Y, \vb G, \vb X, \vb S)\)
by applying Bayes' rule. These updated probabilities may be more
accurate than the original BISG probabilities.

\begin{figure}

\centering{

\includegraphics[width=1\textwidth,height=\textheight]{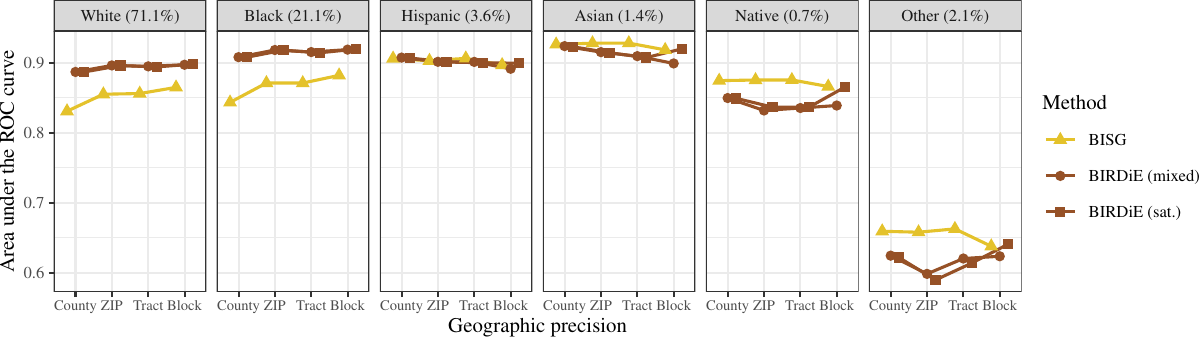}

}

\caption[Race probability predictive accuracy for the input BISG and
BIRDiE-updated probabilities, by race and level of geographic
precision]{\label{fig-nc-roc}Race probability predictive accuracy, as
measured by the area under the receiver operating characteristic (ROC)
curve, for the input BISG probabilities as well as the BIRDiE-updated
probabilities, by race and level of geographic precision. Larger values
indicate more precise predicions.}

\end{figure}%

For example, using the estimates produced by the mixed BIRDiE model
applied to party registration with block-level BISG estimates, the MAP
prediction accuracy increases from 79.6\% with the input probabilities
to 82.9\% with the updated probabilities. These increases are
significantly larger than the differences in accuracy between BISG
probabilities using different levels of geographic precision.

The improvements are reflected in other accuracy measures as well.
Figure~\ref{fig-nc-roc} shows the accuracy of the predictions by race,
as measured by the area under the receiver operating characteristic
(ROC) curve. The updated probabilities are noticeably more accurate than
the input probabilities for White and Black voters, about as accurate
for Hispanic and Asian voters, and slightly less accurate for Native and
``Other'' voters for some geographic levels.

\subsection{Estimates Conditional on an Additional
Variable}\label{estimates-conditional-on-an-additional-variable}

The North Carolina voter file also provides an opportunity to
demonstrate the methodology described in Section~\ref{sec-addlcov} to
produce estimates conditional on another predictor variable that is not
used in the BISG probabilities. We will estimate party registration
rates by race among voters and nonvoters in the 2020 election. Following
the discussion in Section~\ref{sec-addlcov}, we will compute these
estimates two ways: (1) by estimating party registration and 2020
turnout jointly by race, and (2) by first estimating 2020 turnout by
race, then estimating party registration by race and 2020 turnout.

We will use a multinomial mixed-effects BIRDiE model applied to the
block-level BISG probabilities, with random intercepts by tracts, for
all the estimation. Fitting this model to a combined 2020 turnout/party
variable (i.e., one with eight levels: no/DEM, yes/DEM, and so on)
produces estimates of the joint distribution of party registration and
turnout in 2020 by race. Normalizing these probabilities within turnout
and race groups produces estimates of party registration by race and
turnout. The total variation distance between these estimates and the
true distribution is 0.051, which indicates close agreement. The
accuracy of the most-likely race predictions from the BISG probabilities
updated with both party and turnout is about the same as the party-alone
accuracy, indicating that turnout is perhaps less correlated with race
after controlling for location and party registration.

\begin{figure}

\centering{

\includegraphics[width=1\textwidth,height=\textheight]{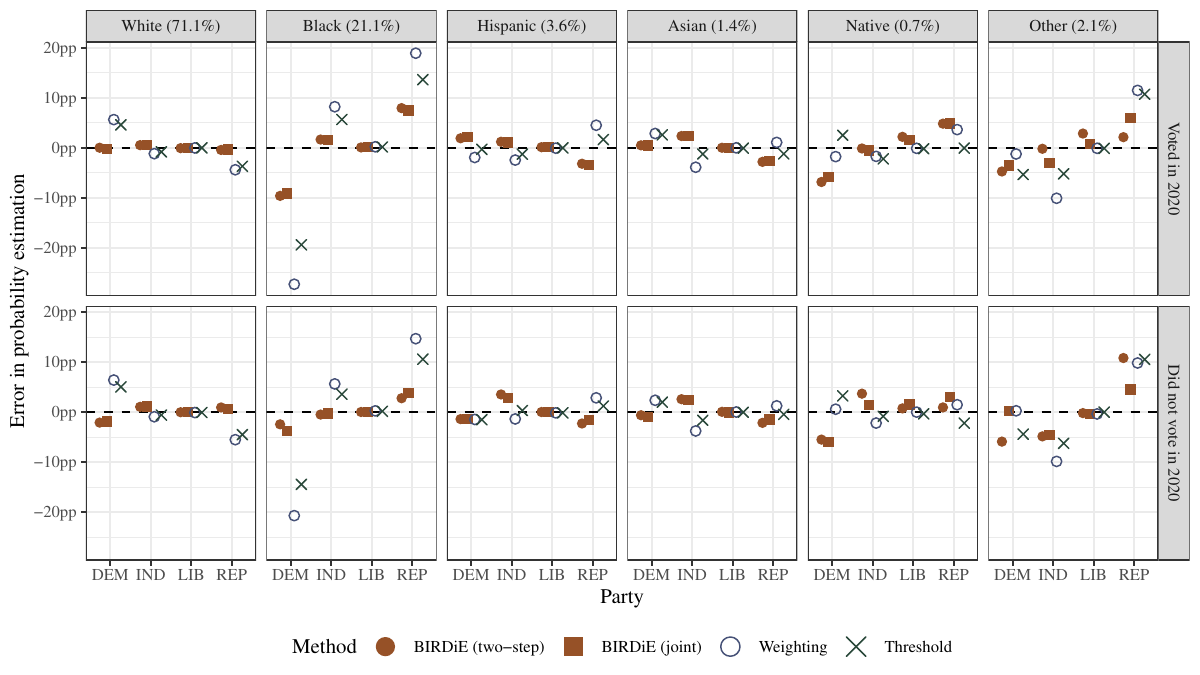}

}

\caption[Error in the conditional expectation estimates for party
registration by turnout and race, by estimation
method]{\label{fig-nc-addlcov}Error in the conditional expectation
estimates for party registration by turnout and race, by estimation
method. All methods used block-level data for this figure. Estimation
uncertainty is minimal and hence suppressed from the figure for
clarity.}

\end{figure}%

Figure~\ref{fig-nc-addlcov} presents the errors in the estimates for
this method (solid square), the two-step method (solid circle), and for
the weighting (open circle) and thresholding (cross) estimators. The TV
distance between these two-stage estimates and the true distribution is
0.052, very similar to the error in the joint-estimation approach. The
two-step approach produces similar estimates to the joint-estimation
approach, though the latter performs better in the ``Other'' category.
As discussed in Section~\ref{sec-addlcov}, while both approaches produce
highly accurate estimates in this example, we would expect the two-stage
approach to be superior when one or both of the variables has more
levels.

In contrast to the BIRDiE models, the weighting and thresholding
estimates of party registration by 2020 turnout and race include large
errors, especially for Black voters with all the errors exceeding 10
percentage points. The TV distance for the weighting estimator is 0.196,
and the distance for the thresholding estimator is 0.147---around 3--4
times higher than for the estimates based on the BIRDiE models.

\subsection{Sensitivity Analysis}\label{sec-nc-sens}

Finally, we examine the sensitivity of our party registration estimates
to potential violations of the key identifying Assumption \ref{a-ci-ys},
following the method outlined in Section~\ref{sec-sens} that is based on
a low-dimensional summary statistic of surnames. We use a publicly
available sample of 5\% of the individual records for the 1930 Census
\citep{cens1930}, which contains individual names, individual and
parental birthplace, and detailed race, ethnicity, and tribal codes.
Since many regions of Asia, particularly Vietnam, experienced little
emigration to the United States before 1930, we further supplement this
data with around 3,000 Asian surnames classified into six regional
subgroups: Chinese, Filipino, Indian, Japanese, Korean, Vietnamese,
NHPI, and Other \citep{asiannames}.

Using these subgroups and the 1930 birthplace and racial data, we can
classify most surnames in the voter file into nine groups (see Appendix
\ref{sec-app-surgrp} for a brief description of the groupings and the
most common 50 surnames for each group). While somewhat arbitrary, these
groups are chosen to combine countries of origin which had significant
immigration to the United States during similar periods.

We first evaluate the plausibility of Assumption \ref{a-ci-ys} by
examining the correlation between the residuals of the BIRDiE model fit
and indicator variables for each of the nine surname groups. Under
Assumption \ref{a-ci-ys}, this residual correlation should be zero
everywhere. As Figure~\ref{fig-nc-sens} shows, however, for many groups
and party labels, the correlation is small but deviates from zero more
than would be expected given only sampling variation. Here, we use the
residuals from the county-level saturated model specification, but the
results are not sensitive to this choice.

\begin{figure}

\centering{

\includegraphics[width=1\textwidth,height=\textheight]{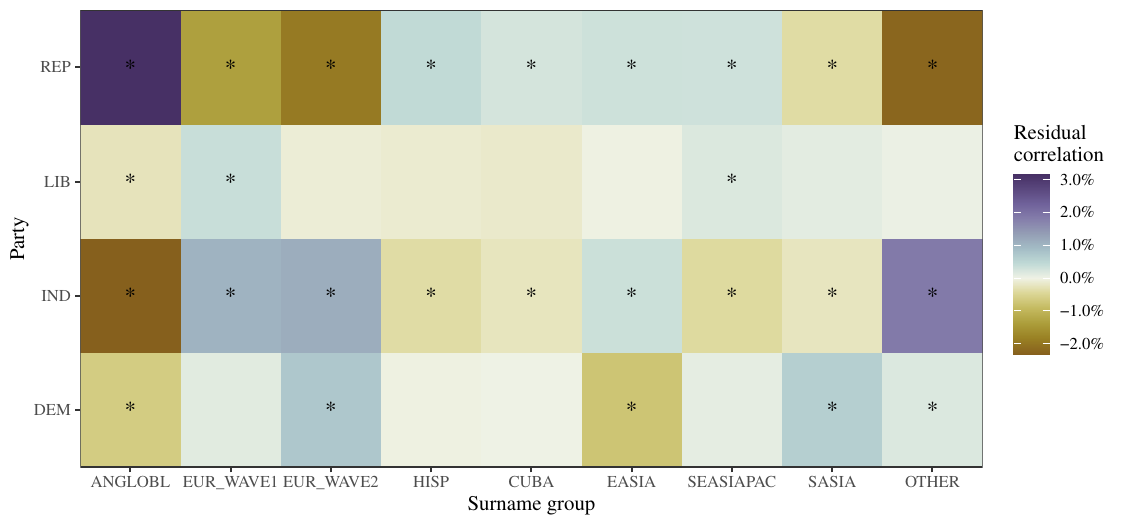}

}

\caption[Residual correlation between party registration and nine
surname groups in North Carolina.]{\label{fig-nc-sens}Residual
correlation between party registration and nine surname groups, after
controlling for race and location. Correlations whose 90\% confidence
intervals exclude zero are marked with an asterisk. See Appendix
\ref{sec-app-surgrp} for details on the surname groups.}

\end{figure}%

Notably, voters with names in the Anglosphere and Black surname group,
which includes surnames that are relatively more common among
many-generation residents of the U.S., such as Smith, Williams, and
Brown, are significantly less likely to register as Democrats and
independents, and more likely to identify as Republicans, even after
controlling for race and geography. Meanwhile, voters with names in the
First and Second wave European immigration surname groups, which include
surnames more common among 19th and 20th century immigrants from Europe,
display the opposite pattern. Differences among surname groups designed
to correlate with membership in various Asian subgroups are also
visible. As might be expected, the relatively many significant
correlations are indicative more of the large number of observations in
the data, rather than large residual correlations themselves---all of
the correlations are quite small in magnitude, with most on the order of
0.01 or so. Thus, we expect our party-by-race estimates to be little
affected by the inclusion of the surname group indicators.

Indeed, re-fitting the county-level saturated model with an additional
surname group covariate produces nearly identical results, albeit at a
moderately higher computational cost given the increased number of
\((G, X, f(S))\) cells. This re-fitting requires Assumption
\ref{a-ci-ysf}, which relaxes Assumption \ref{a-ci-ys}. We find that the
average party registration rate estimate changes only by 0.010, with the
largest change being 0.036, for the rate of Republican registration
among Other voters. The accuracy of the updated BISG probabilities is
likewise virtually unchanged. All in all, this analysis provides
confidence that violations of Assumption \ref{a-ci-ys} for the North
Carolina voter file are likely minor and would have minimal effects on
quantitative and qualitative conclusions.

\section{Application to Tax Data}\label{sec-irs}

We further demonstrate the use of the proposed methodology by applying
it to study racial differences in the home mortgage interest deduction
(HMID) claims, using individual-level tax data from the U.S. Internal
Revenue Service (IRS). This data includes the universe of income tax
returns filed by U.S. taxpayers. The IRS does not collect, and the data
does not include, information on taxpayer race or ethnicity.

\subsection{Home Mortgage Interest Deduction
(HMID)}\label{home-mortgage-interest-deduction-hmid}

The HMID is designed to incentivize home ownership; homeowners with a
mortgage qualify for an itemized deduction based on the amount of
mortgage interest they pay during the year. The deduction is only
available for taxpayers who itemize their deductions. Following
increases to the standard deduction and other tax code changes as part
of the ``Tax Cuts and Jobs Act'' of 2017 (P.L. 115-97), roughly 90\% of
taxpayers take the standard deduction and do not itemize. These
taxpayers are unable to take advantage of the HMID.

The Treasury Department estimates the HMID costs the government about
\$25 billion in foregone revenue in 2019. By budgetary cost, the
deduction is the largest in the income tax code
\citep{crs2017deductions}. Because it is only available to homeowners,
the HMID may disproportionately benefit taxpayers of racial groups that
have a high homeownership rate. Prominent legal scholars have criticized
the possible disproportionate benefits, with
\citet[p.~94]{brown2022whiteness} referring to the subsidy from the HMID
as ``little more than the twenty-first-century version of redlining''
and concluding it ``must be repealed.'' On the other hand, it is also
possible that if Black homeowners faced higher mortgage rates, they
could in principle benefit more than would be expected based on
homeownership rates alone. Lack of administrative data on HMID claims by
race made it difficult to quantify racial disparities that potentially
exist.

The Treasury's internal Office of Tax Analysis has recently used an
extension of the standard BISG model to estimate the usage the HMID and
other deductions by race from individual-level data
\citep{cronin2023bisg}. External researchers have also studied HMID
usage by analyzing survey data or data on proxies like home ownership
\citep{sullivan2017hmid}. Both types of analyses have found that White
taxpayers benefit disproportionately from the HMID, though the magnitude
of the disparity is unclear, especially given the methodological
challenges identified in this paper. Here, we hope to use the BIRDiE
methodology to provide a more precise answer to the question of which
groups are using the HMID and how much they benefit from it.

\subsection{Estimation Procedure}\label{estimation-procedure}

We use a random 10\% sample of individual tax returns (Form 1040s) filed
with the IRS for tax year 2019, a total of 17,145,898 observations. To
calculate individual race probabilities for every observation, we use
the ZIP code tabulation area (ZCTA) corresponding to the geocoded
address listed on the return, plus the last name of the primary filer
using a standard BISG model. This means that conclusions about racial
groups here refer to the race of the primary filer, and not the race of
other household members. For the the roughly 3.4 million records for
which geocoding was not successful, only last names were used.

The outcome variable is the amount of the HMID claimed by the filer,
discretized into 11 levels: one for a deduction of \$0, capturing
roughly 90\% of the sample, and ten levels corresponding to the deciles
of the HMID among those taking the deduction. Given the size of the
data, our outcome model is the relatively simple no-pooling model for
HMID level by geography and racial group. We coarsen the geography
variable used for modeling to the Public Use Microdata Area (PUMA)
level.

Using a larger geography like PUMAs rather than ZCTAs is necessary to
ensure a reasonable amount of data in each outcome-race-geography cell
of the outcome model. PUMAs partition each state into areas containing
roughly 100,000 people. Compared to alternative units of analysis like
states or counties, PUMAs are therefore adaptive to population density;
a large city might be contained entirely in a single county, but with a
dozen or more distinct PUMAs, while the surrounding rural areas might be
spread over many counties but only a few PUMAs. Since we expect more
geographic heterogeneity in and around cities, this feature of PUMAs
lends itself well to our analysis. Where PUMA was not available for the
3.4 million missing geocodes, we used an indicator for state of
residence instead, which is available for every record in the sample.

There are 1,961 distinct PUMAs (or states) in the sample, compared to
28,880 ZCTAs. With 10 non-zero outcome levels and 6 racial groups, the
typical cell in the outcome model is expected to have around 14
observations. Smaller racial groups would have fewer observations than
this, on average. Therefore, to help regularize the PUMA-level model
estimates, we impose a weak empirical Bayes prior based on a global
estimate of HMID by race from the simple weighted estimator. The
effective data size of the prior is just 0.1 observations per racial
group. Thus the prior will have a meaningful impact only in those areas
where there are close to zero expected members of a racial group,
according to the BISG probabilities.

\begin{figure}

\centering{

\includegraphics[width=1\textwidth,height=\textheight]{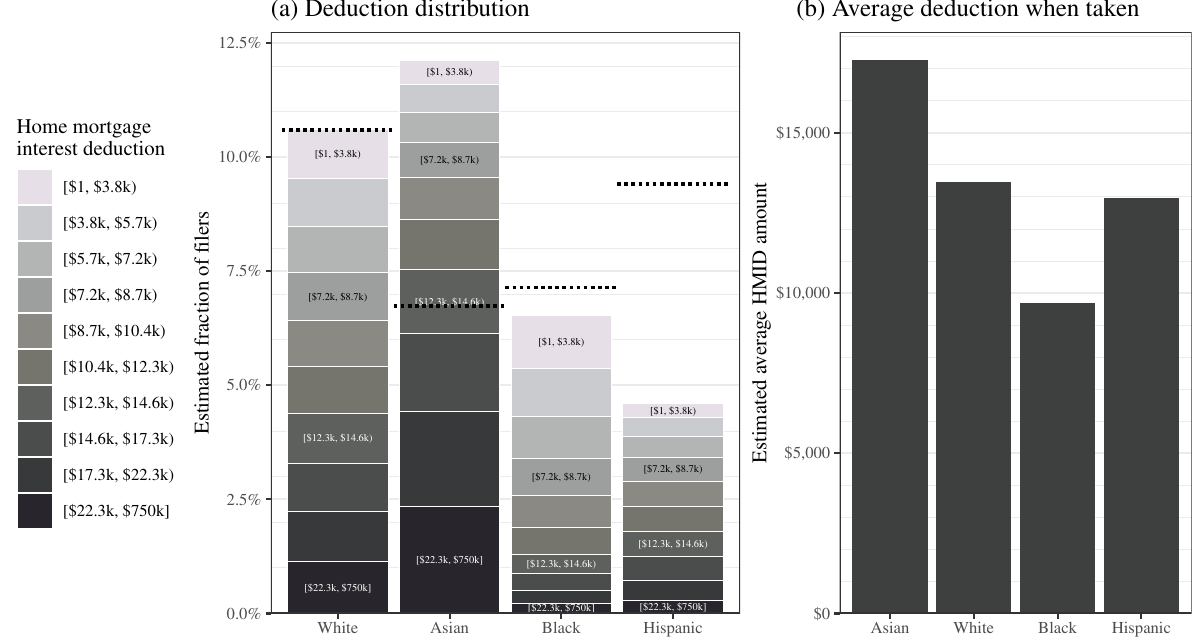}

}

\caption{\label{fig-irs-hmid}Estimates of usage of the home mortgage
interest deduction by race. Panel (a) shows estimates of the average
size of the deduction among filers who took the deduction at all. Panel
(b) shows estimates of the proportion of filers who took the deduction,
further broken into deciles by deduction amount. The dotted lines mark
the expected height of each bar if the only disparities were those in
mortgage rates between racial groups.}

\end{figure}%

\subsection{Findings}\label{findings}

Figure~\ref{fig-irs-hmid}(a) shows the estimated proportion of filers in
each racial group which take the HMID at all, and the distribution of
the HMID claim amount among those who do. Racial disparities are
immediately apparent: while 10.6\% of White filers take the HMID, just
6.5\% of Black filers and 4.6\% of Hispanic filers do. In contrast,
roughly 12\% of Asian filers take the deduction.

How much of this disparity is explained by differences in home ownership
rates across racial groups? We develop estimates of the fraction of each
racial group that has a mortgage for the home they live in based on data
from the 2010 decennial census (see Appendix~\ref{sec-app-mortg} for
details). We then plot as dashed lines in Figure~\ref{fig-irs-hmid}(a)
estimates of the fraction of each racial group that would claim the HMID
if every filer with a mortgage claimed the deduction at the same rate as
White filers do. The figure shows that the lower share of Black
taxpayers claiming the HMID can be explained by the lower home ownership
rate among that group. However, for other groups, disparities remain
after controlling for having a mortgage. In particular, Hispanic filers
claim the HMID at a 2.5 percentage point lower rate than would be
expected based on mortgage rates alone, and Asian filers claim the HMID
at a higher rate than their share of the population with mortgages would
imply. These results suggest that closing disparities in home ownership
may not be sufficient to eliminate disparities in who benefits from the
HMID, as evidently other aspects of filers' situations beyond home
ownership, such as eligibility for other itemized deductions, are
affecting whether their HMID benefits.

\begin{table}[b]

\caption{\label{tab:irs-amt}Estimates of the HMID claim rate and the average deduction by racial group.}
\centering
\begin{tabular}[t]{lllll}
\toprule
Average claim & White & Black & Hispanic & Asian\\
\midrule
Rate & 10.6\% & 6.5\% & 4.6\% & 12\%\\
Among claimants & \$13,500 & \$9,700 & \$13,000 & \$17,300\\
Unconditionally & \$1,400 & \$600 & \$600 & \$2,100\\
\bottomrule
\end{tabular}
\end{table}

Beyond disparities in the rate that taxpayers claim any HMID, there are
also racial differences in the amount of the HMID among claimants. This
is apparent in both the distribution across HMID deciles in
Figure~\ref{fig-irs-hmid}(a) as well as in Figure~\ref{fig-irs-hmid}(b),
which displays estimates of the mean HMID amount among filers who claim
the deduction. These estimates were produced by weighting the observed
HMID amounts according to BIRDiE-updated race probabilities, as
described in Section~\ref{sec-update-bisg}. Compared to White claimants,
whose average deduction is \$13,500, HMID amounts for Black claimants
are skewed toward the lower deciles, translating to a \$2,100 to \$3,800
lower average HMID amount for these groups. In contrast, Hispanic
claimants deduct just \$500 less than White claimants, and Asian
claimants deduct \$3,800 more, on average. Asian filers' deductions fall
into the highest deciles at more than double the rate of any other
group. In fact, a higher fraction of Asian filers take at least a
\$17,000 deduction than the fraction of Hispanic filers who take any
deduction at all. Table \ref{tab:irs-amt} compares these estimates for
claimants to the unconditional averages of the HMID benefit amount
across racial groups.

Overall, our findings support the claims of researchers such as
\citet{moran1996black} and \citet{brown2022whiteness} and others that
the HMID is disproportionately unavailable to Black and Hispanic
taxpayers. Our estimates show that the picture is complicated further by
differences between racial groups even accounting for the prevalence of
mortgages. In addition, the pattern of disparities in overall HMID
claims looks different from the disparities in the amount of the HMID
among claimants.

\section{Discussion}\label{sec-disc}

We have introduced a new identifying assumption and accompanying model,
BIRDiE, and clarified other assumptions implicit in approaches to
disparity estimation when individual race is not observed. In many
real-world applications, we believe that the new model and
identification condition are appropriate and will produce significantly
improved estimates. However, there is no one-size-fits-all approach for
the estimation of racial disparities. For example, the existence of
name-based discrimination may violate our identification assumption
especially when racial categories, for which data are available, are
coarse. Although we provide a sensitivity analysis that partially
addresses this concern, careful consideration of the underlying causal
and information structure is required to avoid making the incorrect
conclusions.

As our empirical demonstration shows, in realistic settings BIRDiE can
substantially outperform existing estimators of racial disparities, both
in aggregate and for small areas. The BIRDiE methodology also produces
improved BISG probabilities, and can be used to estimate disparities
conditional on other variables. These additional features should prove
helpful in practical settings.

\subsection{Recommendations for
Practitioners}\label{recommendations-for-practitioners}

Given the large amount of missing data inherent to the type of racial
disparity estimation studied here, choices about data selection,
processing, and modeling can have a significant impact on estimates and
substantive conclusions. We collect here several recommendations for
practitioners using BISG and BIRDiE methodology in their research.

\begin{itemize}
\item
  \textbf{Choose estimation methodology and input data based on the
  specifics of the research question.} As we have stressed, the causal
  structure of the research setting determines whether BISG should be
  used with the weighting or BIRDiE estimators. Research on populations
  which are very different from the general U.S. population may also
  benefit from the use of alternative or additional data on race and
  geography specific to that population. The choice of whether to use
  state-, county-, ZIP-, tract-, or block-level data likewise depends on
  aspects of the data under study, and the scale of geographic variation
  in the outcome variable and quantities of interest.
\item
  \textbf{Focus on the calibration, not the predictive accuracy, of BISG
  probabilities.} Traditionally, BISG-type methods have been evaluated
  by their predictive accuracy, measured by the mean agreement between
  the thresholded racial categories and ground truth, or the AUROC of
  the BISG probabilities. We find that, especially in cases with
  abundant individual-level data, far more important than maximizing
  predictive accuracy is ensuring that the BISG probabilities are
  properly calibrated. Accurate but biased racial prediction will lead
  to bias in downstream estimates, regardless of which disparity
  estimator is used. To evaluate probabilistic calibration in validation
  settings, we recommend visual diagnostics like binned residual plots,
  as well as numerical summaries like the logarithmic score.
\item
  \textbf{Decide whether additional covariates need to be collected.}
  Additional covariates can make the BISG and BIRDiE assumptions more
  plausible, and can increase the accuracy of BISG probabilities and the
  precision of BIRDiE estimates. We recommend prioritizing covariates
  which are highly predictive of both outcome and race. However,
  including all available covariates without considering their effect on
  the various BISG and BIRDiE assumptions is likely to cause problems.
  In particular, if a covariate's distribution against both race and
  geography is not known, avoid making unrealistic independence
  assumptions that are required for the inclusion of the covariate in
  the BISG model. Rather, follow the approach outlined in
  Section~\ref{sec-addlcov}. Covariate choice should be driven by
  substantive considerations, not convenience.
\item
  \textbf{Unless computational limits are severe, use the mixed-effects
  BIRDiE model with group-level covariates.} The mixed-effects model
  uses the data itself to determine how much to pool disparity estimates
  across different geographies. This avoids the dual pitfalls of
  Simpson's paradox (a risk if no geographical information is used) and
  over-regularization (caused by the prior if no pooling is performed).
  When using the mixed-effects model, group-level covariates are likely
  to improve both overall and especially small-area accuracy, and they
  are easy to collect from public sources. For example, a good default
  for a ZIP-code level model applied to a political outcome variable
  would be to include the percentages of the major racial groups in the
  ZIP code, as well as the ZIP code's income level, population density,
  and partisanship, in the model.
\item
  \textbf{Perform sensitivity analyses.} Both the BISG and BIRDiE models
  use priors, which should be perturbed to examine the sensitivity of
  the results to prior selection. The sensitivity of BIRDiE to its key
  identifying assumption should also be assessed, at minimum using the
  auxiliary covariate approach demonstrated in
  Section~\ref{sec-nc-sens}.
\item
  \textbf{Consider validating estimates with a small-scale survey.} Even
  with administrative microdata observations in the millions, there is
  no substitute for a high-quality random sample of individuals for whom
  race can be observed. Such a sample can be used to validate the
  various assumptions made by BISG and BIRDiE, as well as providing a
  sanity check against BIRDiE disparity estimates. Future improvements
  to BIRDiE could also directly integrate survey data into the workflow.
\end{itemize}

\subsection{Ethical Considerations}\label{ethical-considerations}

As researchers have increasingly applied racial prediction and
imputation methods to administrative records, including many publicly
available datasets, there has been growing concern about ethical and
privacy considerations around performing such predictions. While some
scholars do not view racial prediction methods as privacy risks
\citep[``Statistical Inference is Not a Privacy
Violation'']{buncomment}, we believe it is important for researchers to
consider the implications of their use of racial prediction methods
\citep{kenny2023comment}.

Recently, \citet{lee2023ethical} studied public perception of these
ethical considerations through a large factorial survey experiment that
asked participants if they viewed a hypothetical study as ethically
permissible based on three study factors. They find that studies which
focused on accurate estimation of racial disparities were viewed as more
ethically permissible than studies that overestimated or underestimated
the size of disparities. This tracks with arguments made by many
scholars of progressive tax policy
\citep{brown2022whiteness, bearer2019should}.

Compared to previous approaches to racial disparity estimation, which
emphasized maximizing accuracy of individual racial predictions to
minimize measurement error, BIRDiE focuses on the accuracy of estimating
racial disparities. In fact, as our validation study demonstrated,
different sets of individual race predictions of varying individual
accuracies (county-based versus block-based BISG), when used with
BIRDiE, produced similar estimates of racial disparities. To the extent
that it allows researchers to focus on calibration of racial prediction
rather than maximizing individual predictive accuracy, BIRDiE may
alleviate some privacy concerns and reduce incentives to collect and
link more personal data in an attempt to further increase accuracy. We
view this as a welcome change, consistent with the public's preference
for focusing on accurate disparity information.

However, BIRDiE does allow for the creation of improved BISG
probabilities that incorporate the outcome variable and thus can be more
accurate as well as better calibrated. While this accuracy gain is a
purely statistical phenomenon based on variables already present in the
individual dataset under study, researchers should be cautious, for
example, in releasing these improved racial predictions publicly. We
urge practitioners to view racial prediction tools as a means to the end
of accurate disparity estimation, and treat the intermediate
probabilistic predictions with appropriate care.

\subsection{Future Research}\label{future-research}

Much work remains to be done in accurately and reliably measuring racial
disparities. First, the BISG probabilities themselves can be improved.
Approaches to doing so include \citet{imai2022addressing}, which
accounts for some Census measurement error while remaining
computationally tractable, and \citet{greengard2023bisg}, which rakes
BISG margins to improve calibration. More work on identifying and
producing data sources which can be used as BISG inputs, rather than
relying solely on Census tabulations, will also pay dividends for BISG
quality.

Beyond the BISG probabilities, further empirical analyses could
determine useful additional variables to condition on, which could allow
analysts to weaken the required assumption. Additional study, possibly
combined with qualitative research, could identify causal pathways that
might threaten the assumptions that BISG and BIRDiE rely on, and develop
data sources, like our auxiliary 1930 Census data, that could be used to
evaluate the plausibility of those assumptions in real-world analyses,
and their effect on numerical conclusions. Finally, the BIRDiE model
could also be extended to directly model more complex types of outcome
variables.

\hypertarget{refs}{}

\begin{CSLReferences}{0}{0}\end{CSLReferences}

\appendix

\renewcommand\thefigure{\thesection\arabic{figure}}

\renewcommand\thetable{\thesection\arabic{table}}

\setcounter{figure}{0}

\setcounter{table}{0}

\section{Proofs of Propositions}\label{proofs-of-propositions}

\corunder*

\begin{proof}
For notational simplicity, let $\mu_r=\Pr(Y=y\mid R=r)$ and $\hat\mu_r=\hat\mu^{(\text{wtd})}_{Y|R}(y\mid r)$ for $r\in\{0,1\}$.
Since $\Pr(Y=y\mid R=1, G=g, X=x, S=s)>\Pr(Y=y\mid R=0, G=g, X=x, S=s)$ for all $g\in\cG$, $x\in\cX$, 
necessarily $\E[\Cov(\ind\{Y=y\}, \ind\{R=1\}\mid G,X,S)]>0$ and $\E[\Cov(\ind\{Y=y\}, \ind\{R=1\}\mid G,X,S)]<0$.
We note that the corollary could be stated under this more general condition, but was not for expositional clarity.
Thus by Theorem \ref{thm:wt-bias}, $\hat\mu_1-\mu_r<0$ and $\hat\mu_1-\mu_r>0$.
Then
\begin{align*}
    \hat\mu_1-\hat\mu_0
    &= \hat\mu_1-\mu_1+\mu_1-\mu_0+\mu_0-\hat\mu_0 \\
    &= (\hat\mu_1-\mu_1)-(\hat\mu_0-\mu_0)+(\mu_1-\mu_0) \\
    &< \mu_1-\mu_0,
\end{align*}
as claimed.
\end{proof}

\thmid*

\begin{proof}
Applying the law of total probability and our conditional independence relation $S\indep Y\mid R,G,X$, we have,
for all $y\in\cY$, $g\in\cG$, $x\in\cX$, and $s\in\cS$,
\begin{align*}
    \Pr(Y=y\mid &G=g, X=x, S=s) \\
    &= \sum_{r\in\cR} \Pr(Y=y\mid R=r, G=g, X=x, S=s)\Pr(R=r\mid G=g, X=x, S=s) \\
    &= \sum_{r\in\cR} \Pr(Y=y\mid R=r, G=g, X=x)\Pr(R=r\mid G=g, X=x, S=s).
\end{align*}
The left-hand side is estimable from the data and the rightmost term $\Pr(R=r\mid G=g, X=x, S=s)$ is assumed known.
So for each $y\in\cY$, $g\in\cG$, and $x\in\cX$, this relation is a linear system in unknown parameters $\Pr(Y=y\mid R=r, G=g, X=x)$.
These parameters are identified if and only if this system has a unique solution, i.e. if the coefficient matrix $\vb P$ has rank $|\cR|$ and so does the augmented matrix $\mqty(\vb P&\vb b)$.
\end{proof}

\thmolsunb*

\begin{proof}
Fix $y\in\cY$ and define $m_{gxr} = \E[\ind\{Y=y\}\mid R=r, G=g, X=x)]$.
Then under Assumptions \ref{a-ci-sg}, \ref{a-acc}, and \ref{a-ci-ys},
\begin{align*}
    \E[\ind\{Y=y\}&\mid G=g, X=x, S=s] \\
    &= \sum_{r\in\cR} \E[\ind\{Y=y\}\mid R=r, G=g,X=x)]\Pr(R=r\mid G=g,X=x,S=s) \\
    &= \sum_{r\in\cR} m_{gxr} \hat p_r,
\end{align*}
where as in the main text $\vb{\hat p}$ is the (random) vector of BISG probabilities.
In fact, since the right-hand side depends on $S$  only through $\hat{\vb p}$, we have \[
    \E[\ind\{Y=y\}\mid G=g, X=x, \hat{\vb p}] 
    = \sum_{r\in\cR} m_{gxr} \hat p_r.
\]
So the conditional expectation of $\ind\{Y=y\}$ given $X$, $G$, and the BISG probabilities $\vb{\hat p}$ is linear in those probabilities, with coefficients $m_{gxr}$.
Consequently, the OLS estimate $\hat{\vb*\mu}^{(\text{ols})}_{Y\mid RGX}(y\mid\cdot,g,x)$ will be unbiased for $m_{gxr}$, by the standard results.

Now, we can expand $\Pr(Y=y\mid R=r)$ as
\begin{align*}
    \Pr(Y=y\mid R=r) 
    &= \sum_{x\in\cX,g\in\cG} \Pr(Y=y\mid R=r, G=g, X=x)\Pr(G=g, X=x\mid R=r) \\
    &= \sum_{x\in\cX,g\in\cG} m_{gxr} q_{gx|r}.
\end{align*}
Since $\hat{\vb*\mu}^{(\text{ols})}_{Y\mid RGX}(y\mid \cdot, g,x)$ is unbiased for $m_{gxr}$,  by the linearity of expectation the poststratified estimator $\hat{\mu}^{(\text{p-ols})}_{Y\mid R}(y\mid r)$ is unbiased for $\Pr(Y=y\mid R=r)$.
\end{proof}

\thmwtdols*

\begin{proof}
Fix a $y\in\cY$, $g\in\cG$ and $x\in\cX$.
The weighting estimator of $\Pr(Y=y\mid R=r)$ within the set of individuals with $G_i=g$ and $X_i=x$ may be written $$
    \hat\mu^{(\text{wtd})}_{Y|RGX}(y\mid r,g,x) 
    = \frac{\hat{\vb P}_{\cI(xg) r}^\top \ind\{\vb Y_{\cI(xg)}=y\}}{\hat{\vb P}_{\cI(xg) r}^\top \vb 1}
    = \frac{\norm{\proj_{\hat{\vb P}_{\cI(xg) r}}(\ind\{\vb Y_{\cI(xg)}=y\})}}{
        \norm{\proj_{\hat{\vb P}_{\cI(xg) r}}(\vb 1)}},
$$ the ratio of the projected length of the outcome vector $\ind\{\vb Y=y\}$ and the constant vector $\vb 1$ onto $\hat{\vb P}_{\cdot r}$.
We can write the OLS estimator as $$
    \hat{\vb*\mu}^{(\text{ols})}_{Y|R} = (\hat{\vb P}_{\cI(xg)}^\top\hat{\vb P}_{\cI(xg)})^{-1}
            \hat{\vb P}_{\cI(xg)}^\top \ind\{{\vb Y}_{\cI(xg)}=y\}
    = \mathrm{coord}_{\hat{\vb P}_{\cI(xg)}}(\proj_{\hat{\vb P}_{\cI(xg)}}(\ind\{{\vb Y}_{\cI(xg)}=y\}),
$$ where $\mathrm{coord}_{\hat{\vb P}_{\cI(xg)}}$ is the function that returns the coordinates of its input vector in the $\hat{\vb P}_{\cI(xg)}$ basis (by assumption $\hat{\vb P}_{\cI(xg)}$ has rank $|\cR|$ and so its columns are linearly independent).
To make the comparison even easier, notice that we can break the projection $\proj_{\hat{\vb P}_{\cI(xg) r}}$ into two steps, writing it instead as $\proj_{\hat{\vb P}_{\cI(xg) r}} = \proj_{\hat{\vb P}_{\cI(xg) r}} \circ \proj_{\hat{\vb P}}$.
Letting $\vb Y_{\proj}=\proj_{\hat{\vb P}_{\cI(xg)}}(\ind\{{\vb Y}_{\cI(xg)}=y\})$, then, we can rewrite our estimators as \[
    \hat\mu^{(\text{wtd})}_{Y|R}(y\mid r) 
    = \frac{\norm{\proj_{\hat{\vb P}_{\cI(xg) r}}(\vb Y_\proj)}}{
        \norm{\proj_{\hat{\vb P}_{\cI(xg) r}}(\vb 1)}}  \qand
    \hat{\vb\mu}^{(\text{ols})}_{Y|R}(y\mid r)
    = \mathrm{coord}_{\hat{\vb P}_{\cI(xg)}}(\vb Y_\proj)_r.
\]

Now, since the individual BISG probabilities are nonnegative and sum to 1, a pair $j,k\in\cR$ of races has perfectly discriminating BISG probabilities if and only if the corresponding columns of $\hat{\vb P}$ are orthogonal, i.e., $\hat{\vb P}_{\cI(xg) j}^\top\hat{\vb P}_{\cI(xg) k}=I$.
Begin by writing $\vb Y_\proj$ in terms of the $\hat{\vb P}_{\cI(xg)}$ basis, so \[
    \vb Y_\proj = \sum_{j\in\cR} c_j\hat{\vb P}_{\cI(xg) j},
\] and thus $\hat{\vb\mu}^{(\text{ols})}_{Y|R}(y\mid j)=c_j$.
Without loss of generality, suppose the $c_j$ are numbered as $c_1\ge c_2\ge \cdots\ge c_{|\cR|}$.
We can also expand $\vb 1$ in the same basis. Since the individual probabilities must sum to one, in fact we have \(
    \vb 1 = \sum_{j\in\cR} \hat{\vb P}_{\cI(xg) j}.
\)

For the forward direction, we assume $\hat{\mu}^{(\text{wtd})}_{Y|R}(y\mid j)=\hat{\mu}^{(\text{ols})}_{Y|R}(y\mid j)=c_j$;
multiplying out the denominator of the weighting estimator, we have $\hat{\vb P}_{\cI(xg) j}^\top \vb Y=c_j\hat{\vb P}_{\cI(xg) j}^\top\vb 1$ for all $j$; substituting the basis expansions of $\vb Y_\proj$ and $\vb 1$, this yields \[
    \sum_{k\in\cR} c_k \hat{\vb P}_{\cI(xg) j}^\top\hat{\vb P}_{\cI(xg) k}
    = \sum_{k\in\cR} c_j \hat{\vb P}_{\cI(xg) j}^\top\hat{\vb P}_{\cI(xg) k}, \qq{so}
    \sum_{k\in\cR} (c_j-c_k) \hat{\vb P}_{\cI(xg) j}^\top\hat{\vb P}_{\cI(xg) k}=0.
\] Now fix $j\in J_1=\argmax_j c_j$; this relation still holds, but now every term in the sum is nonnegative and in particular $c_j>c_k$ for all $k\not\in J_1$.
Therefore we must have $\hat{\vb P}_{\cI(xg) j}^\top\hat{\vb P}_{\cI(xg) k}=0$ for all $k\not\in J_1$.
Then fix $j\in J_2=\argmax_{j\not\in J_1} c_j$; since $\hat{\vb P}_{\cI(xg) j}^\top\hat{\vb P}_{\cI(xg) l}=0$ for all $l\in J_1$, every term in the sum is still nonnegative and in particular $c_j>c_k$ for all $k\not\in J_1\cup J_2$.
Therefore we must have $\hat{\vb P}_{\cI(xg) j}^\top\hat{\vb P}_{\cI(xg) k}=0$ for all $k\not\in J_1\cup J_2$.
Proceeding this way through all sets of common values in the $c_j$ we find that for all $j,k\in\cR$, either $c_j=c_k$ or  $\hat{\vb P}_{\cI(xg) j}^\top\hat{\vb P}_{\cI(xg) k}=0$.

For the reverse direction, fix $j\in\cR$ and let $J=\{k\in\cR:c_k=c_j\}$, so that by assumption $\hat{\vb P}_{\cI(xg) j}^\top\hat{\vb P}_{\cI(xg) k}=0$ for all $k\not\in J$.
Then by the above basis expansion, $\hat{\mu}^{(\text{ols})}_{Y|RGX}(y\mid j)=c_j$, and \[
    \hat{\mu}^{(\text{wtd})}_{Y|R}(y\mid j) 
    = \frac{\sum_{k\in\cR} c_k \hat{\vb P}_{\cI(xg) j}^\top\hat{\vb P}_{\cI(xg)k}}{
        \sum_{k\in\cR}\hat{\vb P}_{\cI(xg) j}^\top\hat{\vb P}_{\cI(xg) k}}
    = \frac{c_j \sum_{k\in J} \hat{\vb P}_{\cI(xg) j}^\top\hat{\vb P}_{\cI(xg) k}}{
        \sum_{k\in J}\hat{\vb P}_{\cI(xg) j}^\top\hat{\vb P}_{\cI(xg) k}}
    = c_j = \hat{\mu}^{(\text{ols})}_{Y|R}(y\mid j). \qedhere
\]
\end{proof}

\thmidrel*

\begin{proof}
The argument is identical to the proof of Theorem \ref{thm:id}.

Applying the law of total probability and our conditional independence relation $S\indep Y\mid f(S),R,G,X$, we have,
for all $y\in\cY$, $g\in\cG$, $x\in\cX$, and $s\in\cS$,
\begin{align*}
    \Pr(Y=y\mid G=g, X=x, S=s)
    &= \sum_{r\in\cR} \Pr(Y=y\mid R=r, f(S)=f(s), G=g, X=x, S=s)\\
    &\qquad\times\Pr(R=r\mid G=g, X=x, S=s) \\
    &= \sum_{r\in\cR} \Pr(Y=y\mid R=r, f(S)=f(s), G=g, X=x)\\
    &\qquad\times\Pr(R=r\mid G=g, X=x, S=s).
\end{align*}
The left-hand side is estimable from the data and the rightmost term $\Pr(R=r\mid G=g, X=x, S=s)$ is assumed known.
So for each $y\in\cY$, $z\in f(\cS)$, $g\in\cG$, and $x\in\cX$, this relation is a linear system in unknown parameters $\Pr(Y=y\mid R=r, f(S)=s, G=g, X=x)$.
These parameters are identified if and only if this system has a unique solution, i.e. if the coefficient matrix $\vb P$ has rank $|\cR|$ and so does the augmented matrix $\mqty(\vb P&\vb b)$.
\end{proof}

\section{Additional Small-area Accuracy Evaluation}\label{sec-app-small}

We evaluate the small-area estimates with two additional measures.
First, we calculate the root-mean-square error (RMSE) of the estimated
conditional probabilities by race within each geographic area, and then
average this across all geographic areas. This captures the overall
accuracy of the estimates. Second, to measure how well each method
captures relative differences between geographic areas, we calculate the
correlation between the estimated and true conditional probabilities
across all geographic areas by race. As in the main text, we remove
area-race cells with fewer than 5 voters. A set of estimates which
uniformly underestimates the proportion of Black voters which are
registered Democrats, but which otherwise correctly orders geographic
areas according to their proportion of Black Democrats, will score high
on the correlation measure but also higher in RMSE.

Figure~\ref{fig-nc-smallarea-app} summarizes our results, which closely
track the findings of Section~\ref{sec-small}. The BIRDiE models are
more accurate at all geographic levels and for both Black and White
voters. The weighting estimator performs the worst of all the methods.

\begin{figure}[htb]

\centering{

\includegraphics[width=1\textwidth,height=\textheight]{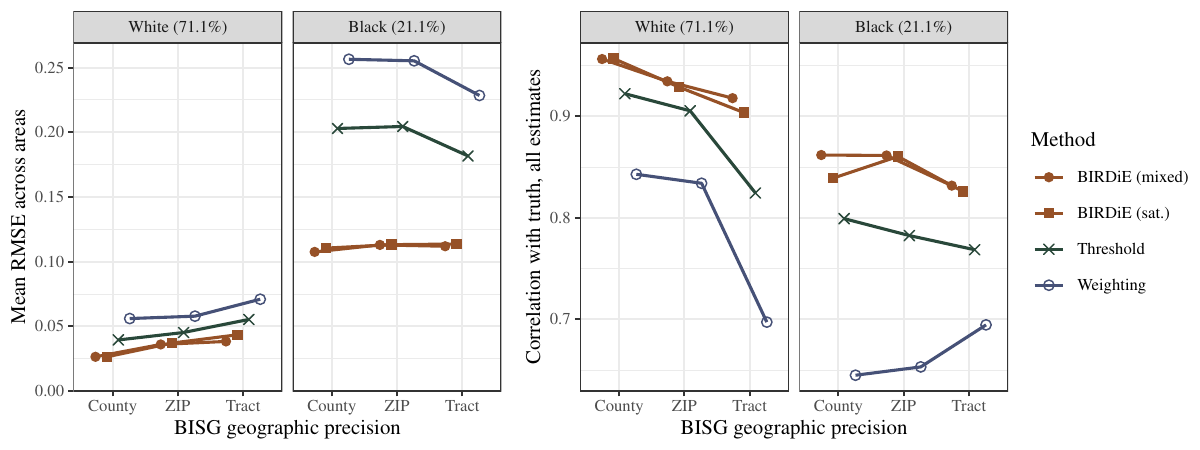}

}

\caption[Alternate measures of accuracy of small-area estimates by
race]{\label{fig-nc-smallarea-app}Accuracy of small-area estimates by
race, as measureed by root-mean-square error (RMSE; lower is better) and
the correlation between the estimates and ground truth (higher is
better).}

\end{figure}%

\section{Sensitivity Analysis}\label{sec-app-sens}

\subsection{Local sensitivity
analysis}\label{local-sensitivity-analysis}

In this section, we develop a sensitivity analysis that assesses how the
bias in BISG race probabilities affect the estimates of racial
disparities. In particular, we consider a setting where Assumptions
\ref{a-ci-sg} and \ref{a-acc} may be violated but Assumption
\ref{a-ci-ys} still holds. For example, consider the existence of
unobserved confounder that affects some or all of the variables except
the outcome, i.e., \((R,S,G,X)\). This leads to the violation of
Assumption \ref{a-ci-sg}, but Assumption \ref{a-ci-ys} continues to be
satisfied so long as such unobserved confounder does not affect the
outcome. Unfortunately, even inaccurate BISG predictions can still lead
to biased estimates of racial disparities.

Specifically, if either the Census data are inaccurate, or the
conditional independence relation does not hold
\(S\notindep G,X\mid R\), then the BISG predictions \(\hat{\vb P}\) will
differ from the ``true'' individual race probabilities
\(\vb P^*=\Pr(R\mid G, X, S)\). Our goal is to quantify how an error in
these probabilities \(\vb P^* - \hat{\vb P}\) shifts the posterior and
hence the estimates of racial disparities.

Denote by \(\pi_{\vb*\delta}\) the posterior constructed using the
error-corrected BISG race probabilities \(\hat{\vb P}_i+\vb*\delta_i\)
as the input probabilities for the model (see Equation
\eqref{eq:posterior}), where \(\pi_{\vb*\delta^*}\) is the true
posterior with \(\vb*\delta_i^* \coloneq \vb P^*_i - \hat{\vb P}_i\).
Estimating how \(\pi\) and \(\pi_{\vb*\delta}\) differ in general is
difficult, but we focus on the settings where \(\vb*\delta\) is small
enough to make a linear approximation appropriate. In sum, we aim to
quantify how the small error in BISG probabilities can alter the
estimates of racial disparities.

For clarity, in this section we will use \(\vb*\theta_{r G_iX_iY_i}\) to
denote the model parameter or function thereof that represents
\(\pi(Y_i\mid R_i=r, G_i, X_i)\). This mirrors the notation of most of
the specific models discussed in Section~\ref{sec-birdie} above. Then
define the following perturbation weight, which represents the ratio of
posterior based on the biased and error-corrected BISG race
probabilities: \[
    w(\Theta,\vb*\delta^*) \coloneq 
    \prod_{i=1}^N 
    \qty(1 + \frac{\vb*\theta_{\cdot G_iX_iY_i}^\top \vb*\delta_i^*}{
    \vb*\theta_{\cdot G_iX_iY_i}^\top \hat{\vb P}_i})
    \propto \frac{\pi_{\vb*\delta^*}(\Theta\mid\vb Y,\vb G,\vb X,\vb S)}{\pi(\Theta\mid\vb Y,\vb G,\vb X,\vb S)},
\] Then, using a local linear approximation, we write the bias for a
particular quantity of interest \(g(\Theta)\) as \begin{align*}
    \E_{\pi_{\vb*\delta^*}}[g(\Theta)] - \E_\pi[g(\Theta)]
    &= \eval{\dv{\E_{\pi_{\vb*\delta}}[g(\Theta)]}{\vb*\delta}}_{\vb*\delta=0}^\top
        \vb*\delta^* + o(\norm{\vb*\delta^*}) \\
    &= \Cov_\pi\qty(g(\Theta), 
        \eval{\dv{\log w(\Theta,\vb*\delta)}{\vb*\delta}}_{\vb*\delta=0})^\top
        \vb*\delta^* + o(\norm{\vb*\delta^*}),
        \numberthis\label{eq:bias-dv}
\end{align*} where the second equality is due to Theorem 2.1 of
\citetext{\citealp{giordano2018cov}; \citealp[see also the idea of
\emph{local sensitivity} from][]{gustafson1996local}}.

With this representation, we can bound the total error in
\(\E_\pi[g(\Theta)]\) for sufficiently small \(\vb*\delta\) as the
following theorem shows.

\begin{theorem}[Bias Bound] \label{thm:bound}
Define $\tilde\vartheta_{ir}\coloneq 
\frac{\theta_{rG_iX_iY_i}}{\theta_{\cdot G_iX_iY_i}^\top\hat{\vb P}_i}$.
Then for any input probabilities with total error $\norm{\vb*\delta^*}^2
=\sum_{i=1}^n\norm{\vb*\delta_i^*}^2\le \Delta^2$,
\begin{equation} \label{eq:covbound}
    |\E_{\pi^*}[g(\Theta)] - \E_\pi[g(\Theta)]|
    \lesssim \Delta\norm{\Cov_\pi(g(\Theta), \tilde\vartheta)},
\end{equation}
as $\Delta\to 0$.
\end{theorem}
\begin{proof}
This is immediate from \eqref{eq:bias-dv} once we compute 
\begin{align*}
    \dv{\log w(\Theta,\vb*\delta)}{\delta_{ir}} 
    &= \dv{\delta_{ir}} \sum_{i=1}^N 
        \log(1+\frac{\vb*\theta_{\cdot G_iX_iY_i}^\top \vb*\delta_i}{
        \vb*\theta_{\cdot G_iX_iY_i}^\top \hat{\vb P}_i}) \\
    &= \dv{\delta_{ir}}
        \log(1+\frac{\vb*\theta_{\cdot G_iX_iY_i}^\top \vb*\delta_i}{
        \vb*\theta_{\cdot G_iX_iY_i}^\top \hat{\vb P}_i}) \\
    &= \frac{1}{1+\frac{\vb*\theta_{\cdot G_iX_iY_i}^\top \vb*\delta_i}{
        \vb*\theta_{\cdot G_iX_iY_i}^\top \hat{\vb P}_i}} \times
        \frac{\vb*\theta_{rG_iX_iY_i}}{
        \vb*\theta_{\cdot G_iX_iY_i}^\top \hat{\vb P}_i} \\
    &= \frac{\vb*\theta_{rG_iX_iY_i}}{
        \vb*\theta_{\cdot G_iX_iY_i}^\top(\hat{\vb P}_i+\vb*\delta_i)}
\end{align*}
and evaluate at $\vb*\delta=0$, since the worst-case bias for a fixed total error can be obtained by having the maximum allowable $\vb*\delta$ point in the direction of the gradient of $\log w(\Theta,\vb*\delta)$.
\end{proof}

The theorem shows that once researchers choose the amount of total error
\(\Delta\), then the bound on the shift in a quantity of interest can be
computed readily from posterior draws. It is important to note that both
\(\vb*\delta^*\) and \(\Cov_\pi(g(\Theta), \tilde\vartheta)\) are
vectors whose dimension depends on the sample size \(N\). Thus, all else
being equal, their norms will each grow as \(\sqrt{N}\). However, each
entry \(\Cov_\pi(g(\Theta), \tilde\vartheta_{ir})\) will tend to shrink
as \(N\) increases, since each observation exerts less leverage on the
overall posterior. Thus the overall impact of the sample size on the
bound in Equation \eqref{eq:covbound} may depend on specific features of
the data. In particular, and as should be expected, the error is not
guaranteed to vanish as \(N\to\infty\). Practitioners should evaluate
Equation \eqref{eq:covbound} under a range of plausible \(\Delta\) to
understand how robust their findings are to worst-case linear violations
of Assumptions \(\ref{a-ci-sg}\) and \(\ref{a-acc}\).

\subsection{OLS sensitivity analysis}\label{ols-sensitivity-analysis}

For a different understanding of the effect of a particular
\(\vb\delta\), we can derive a result on the error in conditional
probability estimates under the OLS estimator and particular
configurations of \(\vb*\delta\). Unlike Theorem \ref{thm:bound}, this
result holds across all sizes of \(\vb*\delta\), and not just
asymptotically as \(\norm{\vb*\delta}\to 0\). However, it applies to the
OLS estimator, which, while unbiased, we do not recommend in practice.
Despite this difference, we expect many of the qualitative conclusions
to hold for BIRDiE models. The effect of any particular \(\vb*\delta\)
can of course be calculated directly by re-fitting the model to new race
probabilities.

Here, we will work with a fixed \(y\in\cY\) and among the subset of
individuals with a particular \(g\in\cG\) and \(x\in\cX\). Then for
notational simplicity we let \(\hat{\vb*\mu}^{\text{(ols)}}\) be the
vector of estimates of \(\Pr(Y=y\mid R, G=g, X=x)\), and \(\vb*\mu\) the
corresponding true probabilities. Similarly, we write \(\hat{\vb P}\)
for the matrix of individual race probability estimates for the subset
of individuals with \(g\in\cG\) and \(x\in\cX\); elsewhere in the text
this would be notated \(\hat{\vb P}_{\cI(xg)}\)

\begin{prop}[OLS Bias from incorrect $\hat{\vb P}$] \label{prop:ols-bias}
Let Assumption \ref{a-ci-ys} hold.
If the OLS estimator $\hat{\vb*\mu}^{\text{(ols)}}$ is 
calculated using race probabilities $\hat{\vb P}$ which differ from the true 
probabilities $\vb P^*=\hat{\vb P}+\vb*\delta$, then its bias satisfies \[
    \E[\hat{\vb*\mu}^{\text{(ols)}}]  - \vb*\mu
    = (\hat{\vb P}^\top \hat{\vb P})^{-1}
    \hat{\vb P}^\top \vb*\delta \vb*\mu
\]
\end{prop}
\begin{proof}
We can write the OLS estimate as \[
    \hat{\vb*\mu}^{\text{(ols)}} 
    = (\hat{\vb P}^\top \hat{\vb P})^{-1} \hat{\vb P}^\top \ind\{\vb Y=y\}.
\]
As shown in Theorem \ref{thm:ols-unb}, under Assumption \ref{a-ci-ys}, 
$\Pr(Y=y\mid S=s,G=g,X=x)$ is linear in the true $\vb P^*$.
Thus letting $\eps=\ind\{\vb Y=y\}-\Pr(Y=y\mid S=s,G=g,X=x)$, we can substitute and find
\begin{align*}
    \hat{\vb*\mu}^{\text{(ols)}} 
    &= (\hat{\vb P}^\top \hat{\vb P})^{-1} \hat{\vb P}^\top 
    \qty(\vb P^*\vb*\mu+\eps) \\
    &= (\hat{\vb P}^\top \hat{\vb P})^{-1} \hat{\vb P}^\top 
    \qty((\hat{\vb P}+\vb*\delta)\vb*\mu+\eps).
\end{align*}
Taking an expectation, since $\E[\eps]=0$ we find
\begin{align*}
    \E[\hat{\vb*\mu}^{\text{(ols)}}]
    &= (\hat{\vb P}^\top \hat{\vb P})^{-1} \hat{\vb P}^\top 
    \qty((\hat{\vb P}+\vb*\delta)\vb*\mu) \\
    &= \vb*\mu + (\hat{\vb P}^\top \hat{\vb P})^{-1} \hat{\vb P}^\top 
    \qty(\vb*\delta\vb*\mu);
\end{align*}
rearrangement yields the result.
\end{proof}

Informally, for BISG error \(\vb*\delta\) to cause problems with the OLS
estimate, two things must happen. First, within individuals it must be
``correlated'' (i.e., have nonzero inner product) with the true
conditional probabilities \(\vb*\mu\). Since \(\vb*\delta_i\) must
always sum to zero, practically, this means that positive BISG errors
must tend to occur in racial groups which have a relatively high
occurrence of outcome \(y\):
\(\Pr(Y=y\mid R=r,G=g,X-x) > \Pr(Y=y\mid G=g,X-x)\). Second, the vector
\(\vb*\delta\vb*\mu\) (where each entry measures this ``correlation''
between errors and relative frequencies) must be correlated the BISG
probabilities themselves \(\hat{\vb P}\). For example, if
\(\vb*\delta_i\vb*\mu\) is positive and tends to be larger for
individuals with a high BISG probability of being Hispanic, then the
overall OLS estimator the conditional probability of \(Y=y\) among
Hispanics will be biased upwards.

While Proposition \ref{prop:ols-bias} applies within a \((G,X)\) cell,
if the same conditions hold across all \((G,X)\) combinations, then the
overall poststratified estimator will be similarly biased.

\section{Surname Groupings for North Carolina Robustness
Analysis}\label{sec-app-surgrp}

We classify every surname in the voter file into one of nine groups,
each containing surnames from one or more of 22 surname groups that we
provide in the replication data and software. These groups are organized
mainly around different regions of the world and different waves of
immigration to the United States. To create the surname groups, each
individual in the 1930 Census data was classified into one of the 22
groups. Then among the set of individuals with each surname, the group
with the highest number of individuals relative to the whole population
was assigned to that surname. For example, while most people named
``Smith'' fall into the Anglosphere group (containing 3rd or more
generation White U.S. residents as of 1930, as well as immigrants from
the U.K., Canada, Australia, etc.), there are relatively more Smiths
among Black people than any other of the 22 groups. Thus ``Smith'' is
assigned to the Black surname group. The full code for creating the 22
surname groupings from the 1930 Census data is available in the
replication materials.

Because of the demographics of the United States, as well as limitations
of the source data there is more geographic specificity in the surname
groupings for some regions (e.g., Europe) than for others (e.g., South
America and Africa). We collapse the 22 surname groups to nine for the
robustness analysis in Section~\ref{sec-valid} based on the demographics
of North Carolina specifically and to minimize the computational burden
of performing the robust analysis.

The 50 roughly most frequent surnames in each group, along with a brief
description of the group, are listed below. We stress that for the
purposes of sensitivity analysis, the surname groups need only be
correlated with countries of origin and racial subgroups. Perfect
alignment is neither possible nor necessary.

\subparagraph{Anlglosphere and Black surname
group.}\label{anlglosphere-and-black-surname-group.}

Surnames which are relatively more prevalent among 3rd-or-more
generation White U.S. residents and Black U.S. residents in 1930.

\nopagebreak

\begin{tabular}{lllll}

1. SMITH & 11. LEWIS & 21. HALL & 31. COLLINS & 41. COX\\
2. WILLIAMS & 12. ROBINSON & 22. CAMPBELL & 32. STEWART & 42. WARD\\
3. BROWN & 13. WALKER & 23. MITCHELL & 33. MORRIS & 43. RICHARDSON\\
4. JONES & 14. ALLEN & 24. CARTER & 34. COOK & 44. WATSON\\
5. DAVIS & 15. WRIGHT & 25. ROBERTS & 35. ROGERS & 45. BROOKS\\
6. TAYLOR & 16. SCOTT & 26. PHILLIPS & 36. MORGAN & 46. WOOD\\
7. MOORE & 17. HILL & 27. EVANS & 37. COOPER & 47. JAMES\\
8. JACKSON & 18. GREEN & 28. TURNER & 38. BAILEY & 48. BENNETT\\
9. WHITE & 19. ADAMS & 29. PARKER & 39. REED & 49. GRAY\\
10. CLARK & 20. BAKER & 30. EDWARDS & 40. HOWARD & 50. HUGHES\\

\end{tabular}

\subparagraph{First wave European immigration surname
group.}\label{first-wave-european-immigration-surname-group.}

Surnames associated with German, Nordic, and Irish immigrants.

\nopagebreak

\begin{tabular}{lllll}

1. JOHNSON & 11. BURNS & 21. CARROLL & 31. SCHULTZ & 41. HIGGINS\\
2. ANDERSON & 12. OLSON & 22. RILEY & 32. PEARSON & 42. OCONNOR\\
3. NELSON & 13. WAGNER & 23. BURKE & 33. BARRETT & 43. QUINN\\
4. MURPHY & 14. MEYER & 24. LARSON & 34. BECK & 44. SWANSON\\
5. PETERSON & 15. SCHMIDT & 25. CARLSON & 35. POWERS & 45. FITZGERALD\\
6. KELLY & 16. RYAN & 26. OBRIEN & 36. LEONARD & 46. CHRISTENSEN\\
7. SULLIVAN & 17. DUNN & 27. LYNCH & 37. BENSON & 47. MANNING\\
8. MURRAY & 18. KELLEY & 28. HANSON & 38. LYONS & 48. MCLAUGHLIN\\
9. MCDONALD & 19. HANSEN & 29. WEBER & 39. MCCARTHY & 49. DOYLE\\
10. KENNEDY & 20. CUNNINGHAM & 30. WALSH & 40. ERICKSON & 50. BRADY\\

\end{tabular}

\subparagraph{Second wave European immigration surname
group.}\label{second-wave-european-immigration-surname-group.}

Surnames associated with Eastern European, Italian, Jewish, Russian,
Greek, and other Southern European immigrants.

\nopagebreak

\begin{tabular}{lllll}

1. FOX & 11. ZIMMERMAN & 21. KLINE & 31. KATZ & 41. NICHOLAS\\
2. NICHOLS & 12. KLEIN & 22. BERGER & 32. MARINO & 42. ROSENBERG\\
3. HOFFMAN & 13. GROSS & 23. STEIN & 33. BRUNO & 43. ROSSI\\
4. NEWMAN & 14. GOODMAN & 24. RAYMOND & 34. MOSER & 44. SINGER\\
5. SCHNEIDER & 15. SHERMAN & 25. FRIEDMAN & 35. GOLDSTEIN & 45. ABRAMS\\
6. KELLER & 16. WOLF & 26. LEVY & 36. GOLDBERG & 46. ACKERMAN\\
7. GREGORY & 17. KRAMER & 27. NOVAK & 37. KAPLAN & 47. HELLER\\
8. SCHWARTZ & 18. NICHOLSON & 28. KAUFMAN & 38. KESSLER & 48. STERN\\
9. COHEN & 19. WEISS & 29. LEVINE & 39. ROMANO & 49. SCHAFER\\
10. BECKER & 20. RUSSO & 30. LEHMAN & 40. FINK & 50. SHAPIRO\\

\end{tabular}

\subparagraph{East Asian surname
group.}\label{east-asian-surname-group.}

Surnames associated with Chinese, Japanese, and Korean immigrants.

\nopagebreak

\begin{tabular}{lllll}

1. LEE & 11. BOWEN & 21. HORNE & 31. HAN & 41. LIANG\\
2. YOUNG & 12. LIU & 22. XIONG & 32. LAU & 42. SUN\\
3. WONG & 13. PAUL & 23. LIM & 33. MA & 43. JUNG\\
4. WANG & 14. CHAN & 24. TANG & 34. PUCKETT & 44. ZHOU\\
5. PARK & 15. TODD & 25. CHO & 35. CHIN & 45. GEE\\
6. MAY & 16. ZHANG & 26. CHENG & 36. GIL & 46. ZHAO\\
7. JOSEPH & 17. LANG & 27. KANG & 37. XU & 47. SHIN\\
8. LOWE & 18. YU & 28. LAW & 38. SONG & 48. OHARA\\
9. CHANG & 19. CHOI & 29. CRAFT & 39. KAY & 49. ZHU\\
10. LIN & 20. MOON & 30. NG & 40. STROUD & 50. YEE\\

\end{tabular}

\subparagraph{South Asian surname
group.}\label{south-asian-surname-group.}

Surnames associated with Indian and Southwest Asian immigrants.

\nopagebreak

\begin{tabular}{lllll}

1. WILSON & 11. CARR & 21. DAVID & 31. MOHAMED & 41. SAMUEL\\
2. THOMAS & 12. SINGH & 22. HOWE & 32. BOGGS & 42. SEWELL\\
3. PATEL & 13. BISHOP & 23. HAHN & 33. KUMAR & 43. HASSAN\\
4. STEVENS & 14. MANN & 24. GOOD & 34. WESTON & 44. SADLER\\
5. WOODS & 15. FRANCIS & 25. JOHN & 35. BEATTY & 45. PINTO\\
6. SHAW & 16. GILL & 26. OSBORN & 36. SWAIN & 46. MAJOR\\
7. FERGUSON & 17. YATES & 27. ABRAHAM & 37. GOMES & 47. BARNHART\\
8. RAY & 18. MARSH & 28. RODRIGUES & 38. JACOB & 48. CARMICHAEL\\
9. WILLIS & 19. ROY & 29. PEREIRA & 39. TOLBERT & 49. MUHAMMAD\\
10. GEORGE & 20. KAUR & 30. SHARMA & 40. PAYTON & 50. GUPTA\\

\end{tabular}

\subparagraph{Southeast Asian and Pacific surname
group.}\label{southeast-asian-and-pacific-surname-group.}

Surnames associated with Southeast Asian and Pacific Islander
immigrants, including Vietnamese and Filipino immigrants.

\nopagebreak

\begin{tabular}{lllll}

1. MILLER & 11. SILVA & 21. HUANG & 31. PHAN & 41. HOANG\\
2. MARTIN & 12. SANTOS & 22. WU & 32. VO & 42. CASH\\
3. KING & 13. GREENE & 23. ROWE & 33. VU & 43. BUI\\
4. NGUYEN & 14. LI & 24. BAUTISTA & 34. LU & 44. CHU\\
5. KIM & 15. LE & 25. HOUSTON & 35. NGO & 45. SINCLAIR\\
6. LONG & 16. YANG & 26. LAM & 36. TAN & 46. SORIANO\\
7. TRAN & 17. LITTLE & 27. HUYNH & 37. HONG & 47. ZHENG\\
8. CHEN & 18. MORAN & 28. HO & 38. DANG & 48. LESLIE\\
9. WEBB & 19. PHAM & 29. CHUNG & 39. DO & 49. ANGEL\\
10. GORDON & 20. RAMSEY & 30. TRUONG & 40. LY & 50. DUONG\\

\end{tabular}

\subparagraph{Non-Cuban Hispanic surname
group.}\label{non-cuban-hispanic-surname-group.}

Surnames associated with Mexican and Latin American immigrants, not
including Cuban immigrants, and Puerto Rican residents.

\nopagebreak

\begin{tabular}{lllll}

1. GARCIA & 11. RIVERA & 21. MENDOZA & 31. CASTRO & 41. ALVARADO\\
2. RODRIGUEZ & 12. GOMEZ & 22. RUIZ & 32. FERNANDEZ & 42. DELGADO\\
3. MARTINEZ & 13. DIAZ & 23. CASTILLO & 33. VARGAS & 43. PENA\\
4. HERNANDEZ & 14. CRUZ & 24. GONZALES & 34. GUZMAN & 44. CONTRERAS\\
5. LOPEZ & 15. REYES & 25. VASQUEZ & 35. MENDEZ & 45. SANDOVAL\\
6. PEREZ & 16. MORALES & 26. ROMERO & 36. MUNOZ & 46. GUERRERO\\
7. SANCHEZ & 17. GUTIERREZ & 27. MORENO & 37. SALAZAR & 47. RIOS\\
8. RAMIREZ & 18. ORTIZ & 28. HERRERA & 38. GARZA & 48. ESTRADA\\
9. TORRES & 19. RAMOS & 29. MEDINA & 39. SOTO & 49. ORTEGA\\
10. FLORES & 20. CHAVEZ & 30. AGUILAR & 40. VAZQUEZ & 50. NUNEZ\\

\end{tabular}

\subparagraph{Cuban surname group.}\label{cuban-surname-group.}

Surnames associated with Cuban immigrants.

\nopagebreak

\begin{tabular}{lllll}

1. GONZALEZ & 11. SUAREZ & 21. CRANE & 31. SARGENT & 41. MARRERO\\
2. ALVAREZ & 12. CONNER & 22. FRYE & 32. GORE & 42. VALDES\\
3. JIMENEZ & 13. SANTANA & 23. PARRA & 33. ZIEGLER & 43. OLIVA\\
4. BOWMAN & 14. DECKER & 24. MAYO & 34. TOMLINSON & 44. MCCLENDON\\
5. DAVIDSON & 15. SKINNER & 25. DAVIES & 35. LOWRY & 45. QUEEN\\
6. ACOSTA & 16. ABBOTT & 26. BLANCO & 36. PAGAN & 46. MCCORD\\
7. MOLINA & 17. GARRISON & 27. WITT & 37. LORD & 47. CRESPO\\
8. MIRANDA & 18. PONCE & 28. CARRASCO & 38. CARBAJAL & 48. CORNEJO\\
9. CASTANEDA & 19. PALACIOS & 29. ALONSO & 39. BETANCOURT & 49. DUMAS\\
10. BALL & 20. SLOAN & 30. HAINES & 40. PATINO & 50. BUENO\\

\end{tabular}

\subparagraph{``Other'' surname group.}\label{other-surname-group.}

Surnames not associated with one of the other categories, including
those associated with later Western European immigration, Middle Eastern
\& North African-associated surnames, Native-associated surnames and
Afro-Caribbean-associated surnames.

\nopagebreak

\begin{tabular}{lllll}

1. PERRY & 11. WELCH & 21. SIMON & 31. FRANCO & 41. MCKENZIE\\
2. HENRY & 12. DAY & 22. CUMMINGS & 32. HAMMOND & 42. BEIL\\
3. HUNT & 13. STANLEY & 23. CHANDLER & 33. CLARKE & 43. COCHRAN\\
4. ROSE & 14. HOPKINS & 24. SHARP & 34. WATERS & 44. NASH\\
5. PIERCE & 15. LAMBERT & 25. BARBER & 35. FRANK & 45. BRYAN\\
6. PETERS & 16. NORRIS & 26. GRIFFITH & 36. ANDRADE & 46. MEYERS\\
7. KNIGHT & 17. WALTERS & 27. PACHECO & 37. LLOYD & 47. CARSON\\
8. RICHARDS & 18. STEELE & 28. CROSS & 38. FRENCH & 48. WILKINSON\\
9. MORRISON & 19. BUSH & 29. GOODWIN & 39. OWEN & 49. ATKINSON\\
10. JACOBS & 20. WOLFE & 30. MULLINS & 40. CHARLES & 50. VINCENT\\

\end{tabular}

\section{Estimating mortgage rates by race}\label{sec-app-mortg}

The decennial census reports the number of homeowners by race, and the
number of homeowners with a mortgage, but does not report the number of
mortgages by race. Thus we are left to infer the mortgage-race
distribution from this marginal information. Fortunately, both
ownership-race and mortgage-ownership marginals are reported at fine
geographic levels.

We therefore produce estimates of mortgage rates by race at the ZCTA
level, then aggregate these estimates to the nation. Stratifying by ZCTA
means that any variation in mortgage rates by race that is explained by
geographic variation will be captured.

Within each ZCTA we estimate the number of mortgages by race by taking
the fraction of homeowners with a mortgage and multiplying by the number
of homeowners of each racial group. Nationwide, we find no association
between the racial composition of a ZCTA and the fraction of homeowners
with a mortgage. While not dispositive because of possible aggregation
bias, this finding nevertheless suggests that after controlling for
geography, little variation by racial group in the fraction of
homeowners with a mortgage remains.

Table \ref{tab:mortg} reports our nationwide estimates of the fraction
of each racial group with a mortgage.

\begin{table}[H]

\caption{\label{tab:mortg}Census-reported number of households, and fraction that own their home, by race, and estimated fraction that have a mortgage, by race.}
\centering
\begin{tabular}[t]{llll}
\toprule
Race & Total households & Fraction owner-occupied & Fraction with mortgage\\
\midrule
White & 82,343,859 & 72.2\% & 49.9\%\\
Black & 13,797,354 & 44.6\% & 31.7\%\\
Hispanic & 14,822,017 & 49.6\% & 33.7\%\\
Asian & 4,580,883 & 58.1\% & 44.4\%\\
Native & 758,975 & 57.2\% & 32.0\%\\
Other & 1,788,360 & 49.3\% & 35.6\%\\
\bottomrule
\end{tabular}
\end{table}


\end{document}